\documentclass[sigconf,screen]{acmart}

\usepackage{balance}

\usepackage{microtype}
\usepackage[ruled, lined, longend, linesnumbered]{algorithm2e}
\usepackage{subcaption}
\usepackage{bbm}
\definecolor{darkblue}{rgb}{0.0,0.0,0.55}

\setcopyright{rightsretained}
\acmPrice{}
\acmDOI{10.1145/3597926.3598129}
\acmYear{2023}
\copyrightyear{2023}
\acmSubmissionID{issta23main-p351-p}
\acmISBN{979-8-4007-0221-1/23/07}
\acmConference[ISSTA '23]{Proceedings of the 32nd ACM SIGSOFT International Symposium on Software Testing and Analysis}{July 17--21, 2023}{Seattle, WA, United States}
\acmBooktitle{Proceedings of the 32nd ACM SIGSOFT International Symposium on Software Testing and Analysis (ISSTA '23), July 17--21, 2023, Seattle, WA, United States}
\received{2023-02-16}
\received[accepted]{2023-05-03}

\begin{document}

\title{Systematic Testing of the Data-Poisoning Robustness of KNN}

\author{Yannan Li}
\affiliation{%
  \institution{University of Southern California}
  \city{Los Angeles}
  \country{United States}}

\author{Jingbo Wang}
\affiliation{%
  \institution{University of Southern California}
  \city{Los Angeles}
  \country{United States}}

\author{Chao Wang}
\affiliation{%
  \institution{University of Southern California}
  \city{Los Angeles}
  \country{United States}}


\begin{abstract}
Data poisoning aims to compromise a machine learning based software component by contaminating its training set to change its prediction results for test inputs. 
Existing methods for deciding data-poisoning robustness have either poor accuracy or long running time and, more importantly, they can only certify some of the truly-robust cases, but remain \emph{inconclusive} when certification fails.  In other words, they cannot falsify the truly-non-robust cases. 
To overcome this limitation, we propose a systematic testing based method, which can falsify as well as certify data-poisoning robustness for a widely used supervised-learning technique named \emph{$k$-nearest neighbors (KNN)}.
Our method is faster and more accurate than the baseline enumeration method, due to a novel over-approximate analysis in the abstract domain, to quickly narrow down the search space, and systematic testing in the concrete domain, to find the actual violations.   
We have evaluated our method on a set of supervised-learning datasets. Our results show that the method significantly outperforms state-of-the-art techniques, and can decide data-poisoning robustness of KNN prediction results for most of the test inputs. 
\end{abstract}



\maketitle

\section{Introduction}
\label{sec:intro}

Testing and verification have always been an integral part of software engineering and, for critical components, rigorous formal analysis techniques are frequently used, either in addition to or together with testing, to ensure that important properties are satisfied.  
With the increasing utilization of machine learning techniques in practical software systems, testing and verification of software components that use machine learning have become important research problems.
Since conventional techniques for testing and verification focus primarily on the software code itself, as opposed to models learned from the data (which are often more important in machine learning based components), there is an urgent need for developing new testing and verification techniques for these emerging software components.

In this paper, we focus on the testing and verification of a security property called \emph{data-poisoning robustness}.
Data poisoning is a type of emerging security risk where the attacker compromises a machine learning based software component by contaminating its training data. 
Specifically, the attacker aims to change the result of a prediction model by injecting a small amount of malicious data into the training set used to learn this model.
Such attacks are possible, for example, when training data elements are collected from online repositories or gathered via crowd-sourcing.  Prior studies have shown the effectiveness of these attacks, e.g., in malware detection systems~\cite{xiao2015feature} and facial recognition systems~\cite{chen2017targeted}.

Faced with such a risk, users may be interested in knowing if the result generated by a \emph{potentially-poisoned} prediction model is still \emph{robust}, i.e., the prediction result remains the same regardless of \emph{whether or how} the training set may have been poisoned by up-to-$n$ data elements~\cite{drews2020proving}. 
This is motivated, for example, by the following use case scenario: the model trainer collects data elements from potentially malicious sources but is confident that the number of potentially-poisoned elements is bounded by $n$; and despite the risk, the model trainer wants to use the learned model to make a prediction for a new test input.  
If we can certify the robustness, the prediction result can still be used; this is called robustness \emph{certification}.
If, on the other hand, we can find a possible scenario that violates the robustness property, the prediction result is discarded; this is called robustness \emph{falsification}.
%
%
Therefore, the robustness falsification and certification problems are analogous to the software testing and verification problems: falsification aims to detect violations of a property, while certification aims to prove that such violations do not exist.

Conceptually, the problem of deciding data-poisoning robustness can be solved as follows.  
First, we assume that the training set $T$ consists of both clean and poisoned data elements, but which of the \emph{up-to-}$n$ data elements are poisoned remains unknown.  
Based on the training set $T$, we use a machine learning algorithm $L$ to obtain a model $M=L(T)$ and then use the model to predict the output class label $y=M(x)$ for a test input $x$. 
Next, we check if the prediction result \emph{could have been different} by removing the poisoned elements from $T$.  
%
Assuming that exactly $1\leq i\leq n$ of the $|T|$ data elements are poisoned, where $n$ is the poisoning threshold, the clean subset $T'\subset T$ will have the remaining $(|T|-i)$ elements.   Using $T'$ to learn the model $M'=L(T')$, we could have predicted the result $y'=M'(x)$. 
Finally, by comparing all of the possible $y'$ with $y$, we decide if prediction for the (unlabeled) test input $x$ is robust: the prediction result is considered robust if and only if,  for all $1\leq i\leq n$,  $y'$ is the same as the default prediction result $y$.

While the solution presented above (called the baseline approach) is a useful mental model,  as an algorithm it is not efficient enough for practical use.  
This is because for a given training set $T$, the number of possible clean subsets ($T'\subset T$) can be as large as $\Sigma_{i=1}^{n} {|T| \choose i}$.  
To see why this is the case, assume that the actual poisoning number $i$ may be any of $1,2,\ldots,n$.  For each specific $i$ value, there are ${|T| \choose i}$ ways of choosing $i$ elements from the $|T|$ elements.  By adding up the numbers for all possible $i$ values, we have $\Sigma_{i=1}^n {|T| \choose i}$.
Due to this combinatorial explosion, it is practically impossible to enumerate all the clean subsets and then check if they all generate the same result as $y=M(x)$.
To avoid the combinatorial explosion, we propose a more efficient method for deciding $n$-poisoning robustness. Instead of enumerating the clean subsets ($T'\subset T$), we use an over-approximate analysis to either verify robustness quickly or narrow down the search space, and in the latter case, rely on systematic testing in the narrowed search space to find a subset $T'$ that can violate robustness.

Our method that combines \emph{quick certification} with \emph{systematic testing} is designed for a supervised learning technique called the  \emph{$k$-nearest neighbors} (KNN) algorithm.
Compared to many other supervised learning techniques, including decision trees and deep neural networks, KNN does not have the high computational cost associated with model training. 
Thus, it has been widely used in software systems to implement classification tasks, including commercial video recommendation systems, document categorization systems, and anomaly detection systems~\cite{guo2003knn,andersson2020predicting,wu2011drex,adeniyi2016automated}. 
KNN is vulnerable to data-poisoning because, in many of these systems, the training data are collected from online repositories or via crowd-sourcing, and thus may be manipulated.

However,  deciding the $n$-poisoning robustness of KNN is a challenging task.
This is because the KNN algorithm has two phases: the learning phase and the prediction phase.  During the learning phase ($K$-parameter tuning phase),  the entire training set $T$ is used to compute the optimal value of parameter $K$ such that, if the most frequent label among the $K$-nearest neighbors of an input is used to generate the prediction label, the average prediction error will be minimized.  Here, the prediction error is computed over data elements in $T$ using a technique called \emph{$p$-fold cross validation} (see Section~\ref{sec:knn}) and the distance used to define nearest neighbors may be the Euclidean distance in the input vector space. 
As a result, the learning phase itself can be time-consuming,  e.g., computing the optimal $K$ for the MNIST dataset with $|T|=$60,000 elements may take 30 minutes, while computing the prediction result for a test input may take less than a minute. 
The large size of $T$ and the complex nature of the mathematical computations make it difficult for conventional software testing and verification techniques to accurately decide the robustness of the KNN system.

To overcome these challenges, we propose three novel techniques. 
First, we propose an over-approximate analysis to certify $n$-poisoning robustness in a sound but incomplete manner.  That is, if the analysis says that the default result $y=M(x)$ is $n$-poisoning robust, the result is guaranteed to be robust. However, this \emph{quick certification} step may return \emph{unknown} and thus is incomplete. 
Second, we propose a search space reduction technique, which analyzes both the learning and the prediction phases of the KNN algorithm in an abstract domain, to extract common properties that all potential robustness violations must satisfy, and then uses these common properties to narrow down the search space in the concrete domain.
Third, we propose a systematic testing technique for the narrowed search space, to find a clean subset $T'\subset T$ that violates the robustness property.  During systematic testing, incremental computation techniques are used to  reduce the computational cost.

We have implemented our method as a software tool that 
takes as input the potentially-poisoned training set $T$, the poisoning threshold $n$, and a test input $x$. The output may be \emph{Certified}, \emph{Falsified} or \emph{Unknown}.  Whenever the output is \emph{Falsified}, a subset $T'\subset T$ is also returned as evidence of the robustness violation. 
We evaluated the tool on a set of benchmarks collected from the literature.  
For comparison, we also applied three alternative approaches.
The first one is the baseline approach that explicitly enumerates all subsets $T'\subset T$. 
The other two are existing methods by Jia et al.~\cite{jia2020certified} and Li et al.~\cite{li2022proving} which only partially solve the robustness problem:  Jia et al.~\cite{jia2020certified} do not analyze the KNN learning phase at all, and thus require the optimal parameter $K$ to be given manually; and both Jia et al.~\cite{jia2020certified} and Li et al.~\cite{li2022proving} focus only on certification in that they may return \emph{Certified} or \emph{Unknown}, but not \emph{Falsified}.  

The benchmarks used in our experimental evaluation are six popular machine learning datasets.  Two of them are small enough that the ground truth (robust or non-robust) may be obtained by the baseline enumerative approach, and thus are useful in evaluating the accuracy of our tool.  The others are larger datasets, e.g., with up to 60,000 training elements and 10,000 test elements, which are useful in evaluating the efficiency of our method. 
The experimental results show that our method can fully decide (either certify or falsify) robustness for the vast majority of test inputs.  

Furthermore, among the four competing methods, our method has the best overall performance. 
Specifically, our method is as accurate as the ground truth (obtained by applying the baseline enumerative approach to small benchmarks) while being significantly faster than the baseline approach. 
Compared with the other two existing methods~\cite{jia2020certified,li2022proving}, our method is significantly more accurate.  For example, on the \emph{CIFAR10} dataset with the poisoning threshold $n=$150, our method successfully resolved 100\% of the test cases, while Li et al.~\cite{li2022proving} resolved only 36.0\%, and Jia et al.~\cite{jia2020certified} resolved only 10.0\%.

To summarize, this paper makes the following contributions:
\begin{itemize}
\item 
We propose the first method capable to \emph{certifying} as well as \emph{falsifying} $n$-poisoning robustness of the entire state-of-the-art KNN system, including both the learning phase and the prediction phase. 
\item
We propose techniques to keep our method accurate as well as efficient, by using over-approximate analysis in the abstract domain to narrow down the search space before using systematic testing to identify violations in the concrete domain.
\item
We implement our method as a software tool and evaluate the tool on six popular supervised-learning datasets to demonstrate the advantages of our method over two state-of-the-art techniques. 
\end{itemize}

The remainder of this paper is organized as follows. 
First, we introduce the technical background in Section~\ref{sec:motivation}.
Then, we present an overview of our method in Section~\ref{sec:alg}, followed by our quick certification subroutine in Section~\ref{sec:quickrobust}, our falsification subroutine in Section~\ref{sec:constraint}, and our incremental computation subroutine in Section~\ref{sec:difflearn}. 
Next, we present the experimental results in Section~\ref{sec:expr}. 
We review the related work in Section~\ref{sec:related}. 
Finally, we give our conclusions in Section~\ref{sec:conclusion}.

\section{Background}
\label{sec:motivation}

In this section, we use two examples to motivate our work and then
highlight the challenges in deciding $n$-poisoning robustness.

\subsection{Two Motivating Examples}

First, let us assume that the potentially-poisoned training set $T$ may be partitioned into $T'$ and $(T\setminus T')$, where $T'$ consists of the clean data elements and $(T\setminus T')$ consists of the poisoned data elements. The KNN's parameter $K$ indicates how many neighbors to consider when predicting the class label for a test input $x$. For example, $K=3$ means that the predicted label of $x$ is the most frequent label among the 3-nearest neighbors of $x$ in the training set.

One of the two ways in which poisoned data may affect the classification result is called \emph{direct influence}. In this case,  the poisoned elements directly change the $K$-nearest neighbors of $x$ and thus the most frequent label, as shown in Figure~\ref{fig:ex1}.

Figure~\ref{fig:ex1}(a) shows only the clean subset $T'$, where the \emph{triangle}s and \emph{star}s represent the training data elements, and the \emph{square} represents the test input $x$.  Furthermore, \emph{triangle} and \emph{star} represent the two distinct output class labels.
The goal is to predict the output class label of the test input $x$. In this figure, the dashed circle contains the 3-nearest neighbors of $x$.  Since the most frequent label is \emph{star},  $x$ is classified as \emph{star}.

Figure~\ref{fig:ex1}(b) shows the entire training set $T$, including all of the elements in $T'$ as well as a poisoned data element.   In this figure, the dashed circle contains the 3-nearest neighbors of $x$.  Due to the poisoned data element, the most frequent label becomes \emph{triangle} and, as a result, $x$ is  mistakenly classified as \emph{triangle}.

The other way in which poisoned data may affect the classification result is called \emph{indirect influence}.  In this case, the poisoned elements may not be close neighbors of $x$, but their presence in $T$ changes the parameter $K$ (Section \ref{sec:knn} explains how to compute K), and thus the prediction label.

Figure~\ref{fig:ex2} shows such an example where the poisoned element is not one of the 3-nearest neighbors of $x$.  However,  its presence changes the parameter $K$ from 3 to 5 in Figure~\ref{fig:ex2}(b).  As a result,  the predicted label for $x$ is changed from \emph{star} in Figure~\ref{fig:ex2}(a) to \emph{triangle} in Figure~\ref{fig:ex2}(b).

The existence of \emph{indirect influence} prevents us from verifying robustness by only considering the cases where poisoned elements are near $x$ (which is the unsound approach of Jia et al.~\cite{jia2020certified}); instead, we must consider  each $T' \in \Delta_n(T)$.

\begin{figure}
 \centering
	\begin{subfigure}[b]{0.22\textwidth}
	 \centering
	\includegraphics[width=.9\textwidth]{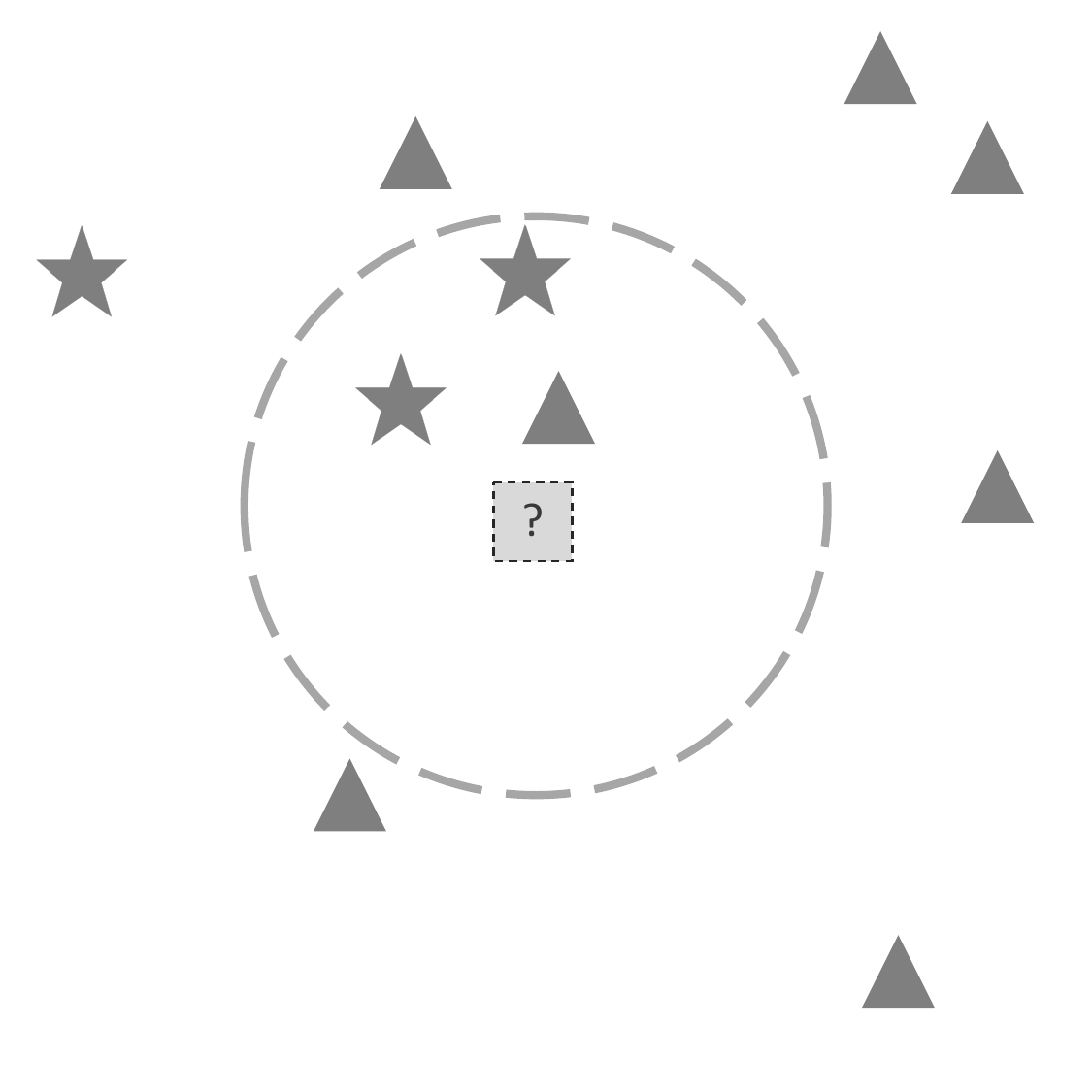}
	\caption{\scriptsize clean subset $T'$ ($K$=3)}
	\end{subfigure}
\hfill
	\begin{subfigure}[b]{0.22\textwidth}
	 \centering
	\includegraphics[width=.9\textwidth]{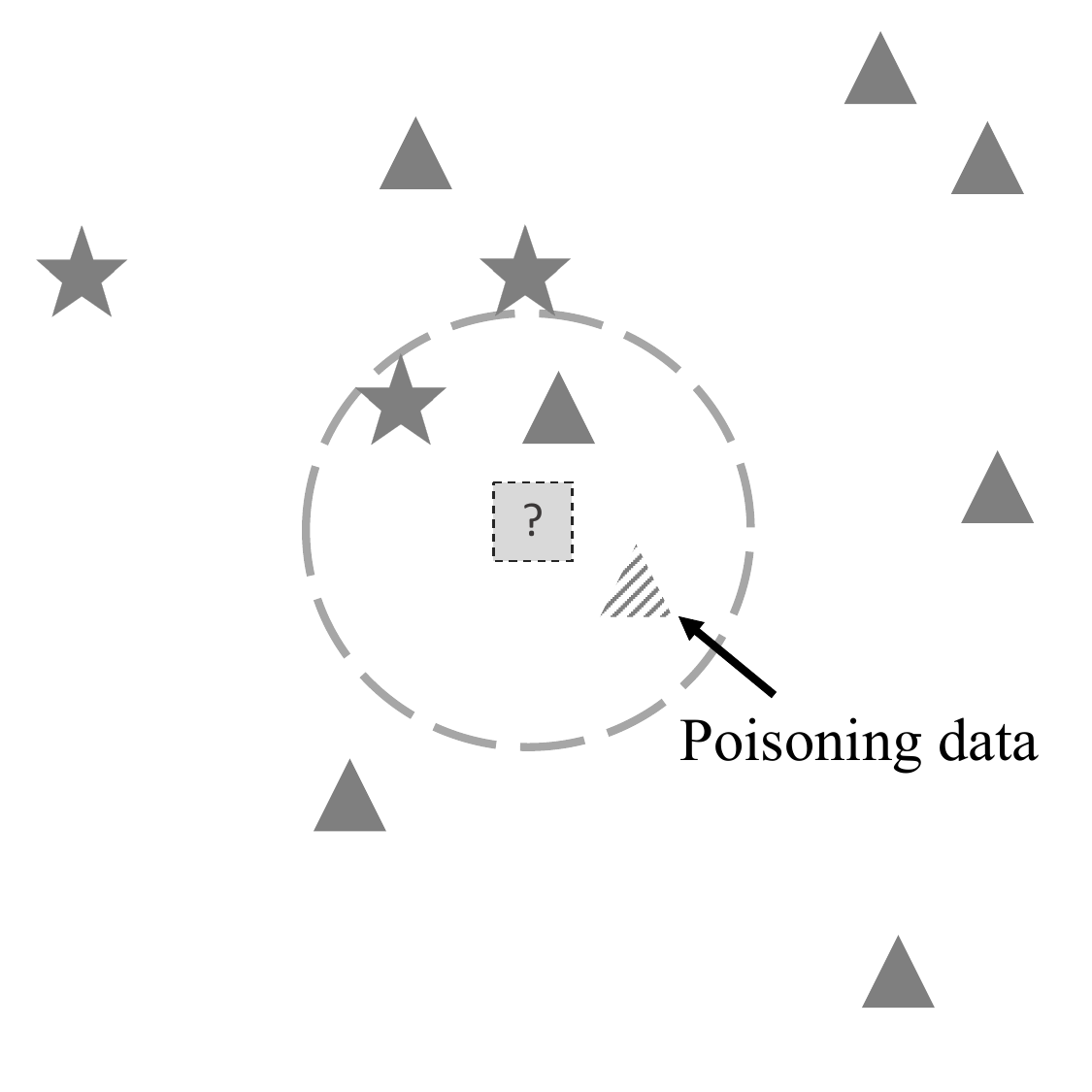}
	\caption{\scriptsize poisoned set $T$ ($K$=3)}
	\end{subfigure}
\caption{Example of \emph{direct influence} by poisoning data.}
\label{fig:ex1}
\end{figure}

\begin{figure}
 \centering
	\begin{subfigure}[b]{0.22\textwidth}
	 \centering
	\includegraphics[width=.9\textwidth]{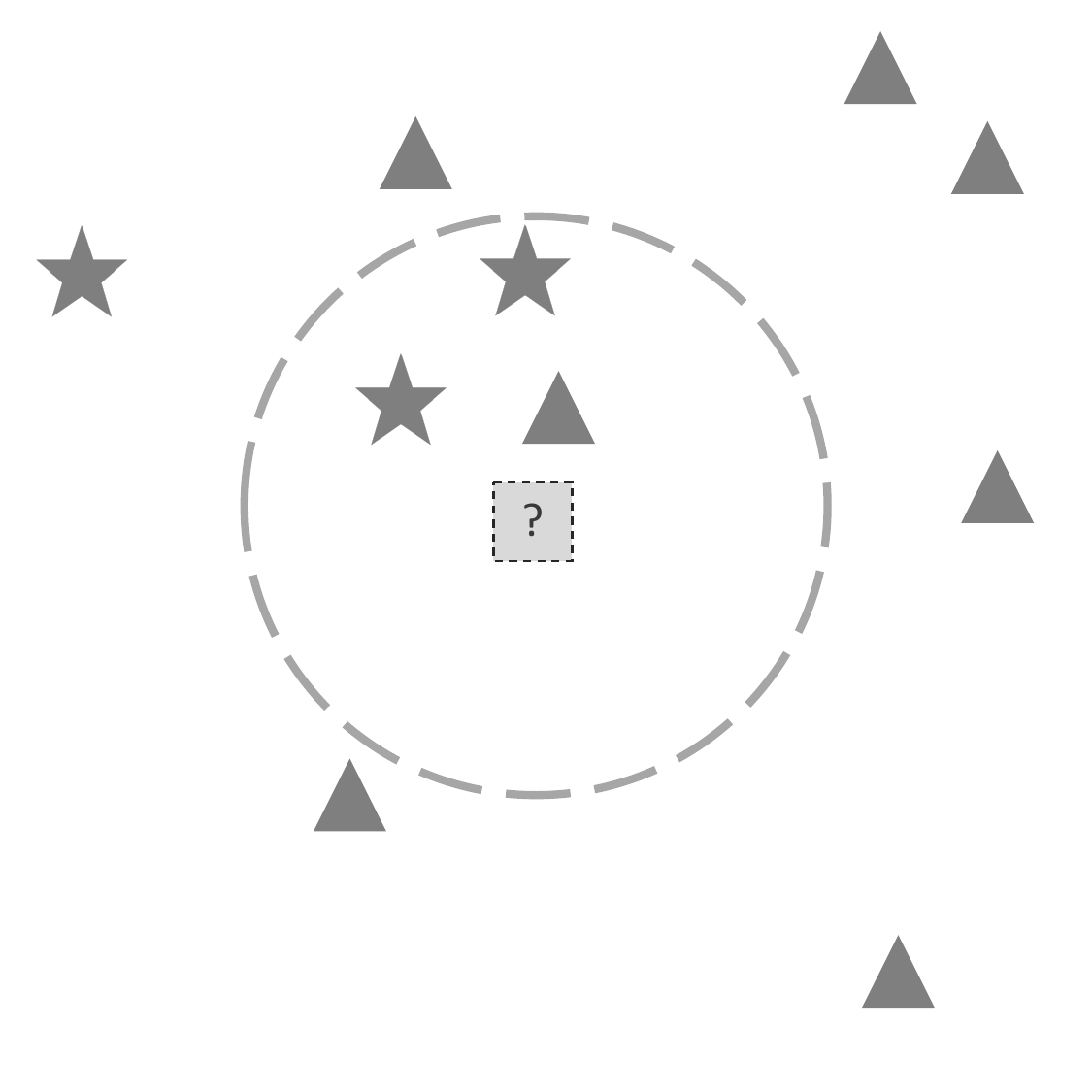}
	\caption{\scriptsize clean subset $T'$ ($K$=3)}
	\end{subfigure}
\hfill
	\begin{subfigure}[b]{0.22\textwidth}
	 \centering
	\includegraphics[width=.9\textwidth]{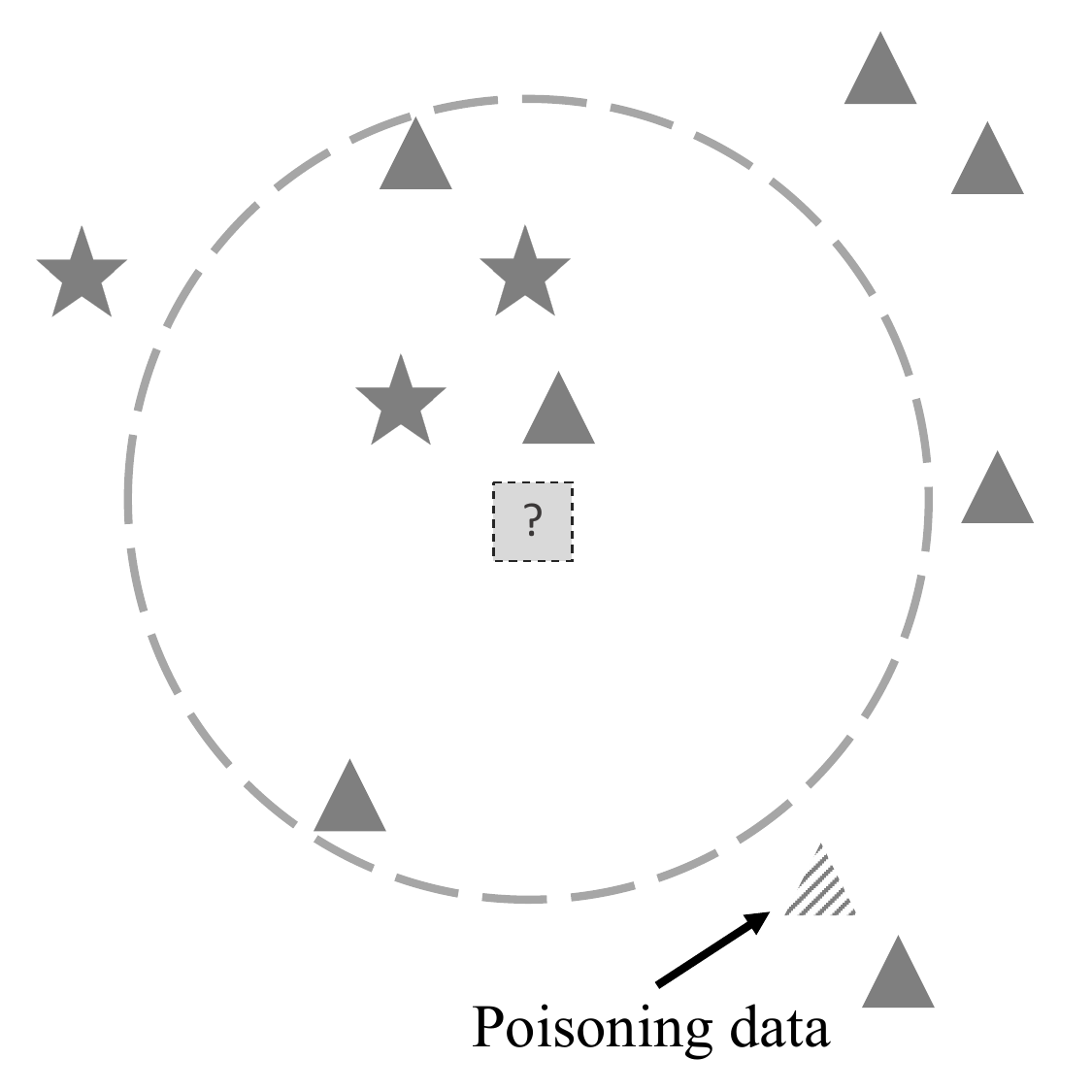}
	\caption{\scriptsize poisoned set $T$ ($K$=5)}
	\end{subfigure}
\caption{Example of \emph{indirect influence} by poisoning data.}
\label{fig:ex2}
\end{figure}

\subsection{The $k$-Nearest Neighbors (KNN)}
\label{sec:knn}

Let $L$ be a learning algorithm, $M=L(T)$, which takes a set $T = \{(x,y)\}$ of labeled elements as input and returns a model $M$ as output.  
Inside $T$, each $x\in \mathcal{X} \subseteq \mathbb{R}^D$ is a vector in the $D$-dimensional input feature space $\mathcal{X}$, and each $y\in \mathcal{Y} \subseteq \mathbb{N}$ is a class label in the output label space $\mathcal{Y}$.
The model is a function $M : \mathcal{X} \rightarrow \mathcal{Y}$ that maps a test input $x\in\mathcal{X}$ to a class label $y\in\mathcal{Y}$. 

The KNN algorithm consists of two phases.  In the learning phase, the labeled data in $T$ are used to compute the optimal value of the parameter $K$. In the predication phase,  an unlabeled input $x\in \mathcal{X}$ is classified as the most frequent label among the $K$ nearest neighbors of $x$ in $T$.
The distance used to decide $x$'s neighbors in $T$ may be measured using several metrics. In this work, we use the most widely adopted Euclidean distance in the input feature space $\mathcal{X}$.  

%
%
To compute the optimal $K$ value, state-of-the-art KNN implementations iterate through all possible candidate values in a reasonable range, e.g., $1\sim 5000$, and use a technique called \emph{$p$-fold cross validation} to identify the optimal value.
The optimal $K$ value is the one that has the smallest average prediction error.
During $p$-fold cross validation, $T$ is randomly divided into $p$ groups of approximately equal size. Then, for each candidate $K$ value, the prediction error of each group is computed, by treating this group as a test set and the union of all the other $p-1$ groups as the training set. 
Finally, the prediction errors of the individual groups are used to compute the average prediction error among all $p$ groups.

\subsection{The $n$-Poisoning Robustness}

We follow the definition given by Drews et al.~\cite{drews2020proving}, which was introduced initially for models such as decision tree~\cite{MeyerAD21} and linear regression~\cite{abs-2206-03575} but was also applied to KNN~\cite{li2022proving}.
It has a significant advantage: the definition can be applied to \emph{unlabeled} data, since robustness does not depend on the actual label of the test input $x$.  This is important because the actual label of the test input (i.e., the ground truth) is often unknown in practice. 

Given a potentially-poisoned training set $T$ and a poisoning threshold $n$ indicating the maximal poisoning count, the set of possible clean subsets of $T$ is represented by
$\Delta_n(T)=\{T' \mid T'\subset T \mbox{ and } |T\setminus T'|\leq n\}$.  That is, $\Delta_n(T)$ captures all possible situations where the poisoned elements are eliminated from $T$.

We say the prediction $y=M(x)$ for a test input $x$ is robust if and only, for all $T'\in\Delta_n(T)$ such that $M'=L(T')$ and $y'=M'(x)$, we have $y'=y$.
In other words, the default result $y=M(x)$ is the same as all of the possible results, $y'=M'(x)$, no matter which are the  ($i\leq n$) poisoned data elements  in the training set $T$.

\subsection{The Baseline Method}

We first present the \emph{baseline} method in Algorithm~\ref{alg:baseline}, and then compare it with our proposed method in Algorithm~\ref{alg:our_method} (Section~\ref{sec:alg}).

\begin{algorithm}[t]
\caption{\makebox{Procedure \textsc{Falsify\_Baseline}$(T, n, x)$.}}
\label{alg:baseline}
{\footnotesize
$K \gets \textsc{KNN\_learn}(T)$\\
$y \gets \textsc{KNN\_predict}(T,K, x)$ \\
$\Delta_n(T) \gets \{T' \mid T'\subset T \mbox{ and } |T\setminus T'|\leq n\}$\\
\While {$\Delta_n(T) \neq \emptyset$ $\wedge$ consumed\_time $<$ time\_limit} {
        Remove a clean subset $T'$ from  $\Delta_n(T)$\\
        $K' \gets  \textsc{KNN\_learn}(T')$\\
        $y' \gets  \textsc{KNN\_predict}(T',K', x)$\\
        \If {$y \neq y'$} { \Return \emph{Falsified} with ($T \setminus T'$) as evidence}
        
}
\eIf {$\Delta_n(T) = \emptyset$} { \Return \emph{Certified}}
{\Return \emph{Unknown}}
}
\end{algorithm}

The baseline method explicitly enumerates the possible clean subsets $T'\in\Delta_n(T)$ to check if the prediction result $y'$ produced by $T'$ is the same as the prediction result $y$ produced by $T$ for the given input $x$. 
As shown in Algorithm~\ref{alg:baseline}, the input consists of the training set $T$, the poisoning threshold $n$, and the test input $x$.  The subroutines \textsc{KNN\_learn} and \textsc{KNN\_predict} implement the \emph{standard} learning and prediction phases of the KNN algorithm.
Without the time limit, the baseline method would be both sound and complete; in other words, 
it would return either \emph{Certified} (Line~13) or \emph{Falsified} (Line~9).
With the time limit, however, the baseline method  will return \emph{Unknown} (Line~15) after it times out.

The baseline procedure is inefficient for three reasons.
First, it is a \emph{slow certification} (Line~13) to check whether the prediction result for $x$ remains the same for all possible clean subsets $T'\in\Delta_n(T)$. In many cases, the elements around $x$ are almost all from one class, and thus $x$'s predicted label cannot be changed by either direct or indirect influence.  However, the baseline procedure cannot quickly identify and exploit this to avoid enumeration. 
Second, even if a violating subset $T'$ exists, the vast majority of subsets in $\Delta_n(T)$ are often non-violating.  However, the baseline procedure cannot quickly identify the violating $T'$ from $\Delta_n(T)$.
Third, within the while-loop, different subsets share common computations inside \textsc{KNN\_learn}, but these common computations are not leveraged by the baseline procedure to reduce the computational cost.

\section {Overview of The Proposed Method}
\label{sec:alg}

There are three main differences between our method in Algorithm~\ref{alg:our_method} and the baseline method in Algorithm~\ref{alg:baseline}. They are marked in dark blue.  They are the novel components designed specifically to overcome limitations of the baseline method.

First, we add the subroutine \textsc{QuickCertify} to quickly check whether it is possible to change the prediction result for the test input $x$. This is a sound but incomplete check in that, if the subroutine succeeds, we guarantee that the result is robust.  If it fails, however, the result remains unknown and we still need to execute the rest of the procedure. The detailed implementation of \textsc{QuickCertify} is presented in Section~\ref{sec:quickrobust}.

Second, before searching for a clean subset that violates robustness, we compute $\nabla_n^x(T) \subseteq \Delta_n(T)$, to capture the \emph{likely violating} subsets.  In other words, the \emph{obviously non-violating} ones in $\Delta_n(T)$ are safely skipped.  Note that, while $\Delta_n(T)$ depends only on $T$ and $n$, $\nabla_n^x(T)$ depends also on the test input $x$.  For this reason, $\nabla_n^x(T)$ is expected to be significantly smaller than $\Delta_n(T)$, thus reducing the search space.   The detailed implementation of \textsc{GenPromisingSubsets} is presented in Section~\ref{sec:constraint}.

Third, instead of applying the standard \textsc{KNN\_learn} subroutine to each subset $T'$ to perform the expensive $p$-fold cross validation,  we split it to \textsc{KNN\_learn\_init} and \textsc{KNN\_learn\_update}, where the first subroutine is applied only once to the original training set $T$,  and the second subroutine is applied to each subset $T'\in\nabla_n^x(T)$.  Within \textsc{KNN\_learn\_update}, instead of performing $p$-fold cross validation for $T'$ from scratch,  we leverage the results returned by \textsc{KNN\_learn\_init} to incrementally compute the results for $K'$.  The detailed implementation of these two new subroutines is presented in Section~\ref{sec:difflearn}.

\begin{algorithm}[t]
\caption{\makebox{Our new procedure \textsc{Falsify\_New}$(T, n, x)$.}}
\label{alg:our_method}
{\footnotesize
\If {$\textcolor{darkblue}{\textsc{QuickCertify}(T, n, x)}$} {\Return \emph{Certified} }
$\langle K, \textcolor{darkblue}{Error} \rangle \gets \textcolor{darkblue}{\textsc{KNN\_learn\_init}}(T)$\\
$ y \gets  \textsc{KNN\_predict}(T, K, x)$\\
$\nabla_n^x(T) \gets  \textcolor{darkblue}{\textsc{GenPromisingSubsets}}(T, n, \textcolor{darkblue}{ x, y})$\\
\While {$\nabla_n^x(T) \neq \emptyset$ $\wedge$ consumed\_time $<$ time\_limit} {
	Remove a subset $T'$ from $\nabla_n^x(T)$\\
	$K' \gets \textcolor{darkblue}{\textsc{KNN\_learn\_update}}(T \setminus T', \textcolor{darkblue}{Error})$\\
	$y' \gets \textsc{KNN\_predict}(T', K', x)$\\
	\If {$y \neq y'$} { \Return \emph{Falsified} with ($T \setminus T'$) as evidence}
}
\eIf {$\nabla_n^x(T) = \emptyset$} { \Return \emph{Certified} }
{\Return \emph{Unknown} }
}
\end{algorithm}

\textcolor{black}{
To summarize, our method first uses over-approximation to certify robustness. If it succeeds, the classification result is guaranteed to be robust; otherwise, the classification result remains unknown. Only for the unknown case, our method uses under-approximation to falsify robustness. If it succeeds, the classification result is guaranteed to be not robust. Otherwise, the classification result remains unknown.
Therefore, our method does not ``mix'' over- and under-approximations in the sense that they are never used simultaneously; instead, over- and under-approximations are used sequentially in two separate steps of our algorithm.
The formal guarantee is that: If our method says that a case is robust, it is indeed robust (see Theorem~\ref{thm:quick}); if our method says that a case is not robust, it is indeed not robust (since a poisoning set is found); and if our method says unknown, it may be either robust or not robust.
}

\section{Quickly Certifying Robustness}
\label{sec:quickrobust}

In this section, we present the subroutine \emph{QuickCertify}, which is a \emph{sound but incomplete} procedure for certifying robustness of the KNN for a given input $x$.  Therefore, if it returns \texttt{True}, the prediction result for $x$ is guaranteed to be robust.  If it returns \texttt{False}, however, we still need further investigation.

\begin{table}
\centering
\caption{Notations used in our new algorithm.}
\label{tbl:notations}
\scalebox{0.93}{
\footnotesize
\begin{tabular}{p{1.75cm}p{6.5cm}}
  \hline

\textbf{Training Set $T$} & 
Let $T = \{(x_1, y_1), (x_2, y_2),$ ..., $(x_m,y_m)\}$ be a set of labeled data elements, where input $x_i \in \mathcal{X} \subseteq \mathbb{R}^D$ is a feature vector in the feature space $\mathcal{X}$, and $y\in \mathcal{Y} \subseteq \mathbb{N}$ is a class label in the label space $\mathcal{Y}$.
\\\hline

\textbf{Set of $K$-nearest Neighbors $T_x^K$} &
Let $T_x^K$ be the set of $K$ nearest neighbors of test input $x$ in the training set $T$. 
\\
& \\\hline

\textbf{Label Counter $\mathcal{E}(\cdot)$}&
Let $\mathcal{E}(D) = \{~(l_i:\#l_i)~ \}$ be the set of label counts for a dataset $D$, where $l_i \in \mathbb{Y}$ is a label and $\#l_i \in \mathbb{N}$ is the number of elements in $D$ with label $l_i$.
\\
& \\\hline

\textbf{Most Frequent Label $Freq(\cdot)$} & 
Let $Freq(\mathcal{E}(D))$ be the most frequent label in the label counter $\mathcal{E}(D)$ for the dataset $D$. 
\\
& \\\hline

\end{tabular}
}
\end{table}

We define the notations used by the KNN algorithm in Table~\ref{tbl:notations}, following the ones used by Li et al.~\cite{li2022proving}.
%
%
Consider $T_x^3 = \{(x_1, l_a), (x_2, l_a), (x_3, l_b)\}$ as an example, which captures the 3-nearest neighbors of a test input $x$.  Then the corresponding label counter is $\mathcal{E}(T_x^3) = \{(l_a: 2), (l_b: 1)\}$, meaning that two elements in $T_x^3$ have the label $l_a$ and one element has  the label $l_b$. The corresponding most frequent label is $Freq(\mathcal{E}(T_x^3)) = l_a$.

%
For each subset $T'\in\Delta_n(T)$, we define a removal set $R=(T\setminus T')$ and a removal strategy $\mathcal{S}=\mathcal{E}(R)$.  

\begin{itemize}
\item 

A \emph{removal set} $R$ for a set $T$ is a non-empty subset $R \subset T$, to represent the removal of the elements in $R$ from $T$. 
\item

A \emph{removal strategy} $\mathcal{S}$ is the label counter of a removal set $R$, i.e.,  $\mathcal{S} = \mathcal{E}(R)$.  

%

\end{itemize}
Thus, all the removal sets form the \emph{concrete domain}, and all the removal strategies form an \emph{abstract domain}. While analysis in the (large) concrete domain is expensive, analysis in the (smaller) abstract domain is much cheaper.  This is analogous to the \emph{abstract interpretation}~\cite{CousotC77} paradigm for static program analysis\footnote{
\textcolor{black}{
There are Galois connections~\cite{CousotC14} $(\alpha,\gamma)$ between removal sets and removal strategies (multisets) that are standard in the context of abstract interpretation, where the $\alpha$ function abstracts removal sets in the concrete domain to removal strategies (multisets) in the abstract domain, and the $\gamma$ function concretizes the multisets back to sets.
}
}.

For the set $T_x^3$ above, there are 6 removal sets: 
\makebox{$R_1=\{(x_1, l_a)\}$}, \makebox{$R_2=\{(x_2, l_a)\}$}, \makebox{$R_3=\{(x_3, l_b)\}$}, \makebox{$R_4=\{(x_1, l_a), (x_2, l_a)\}$},  $R_5=\{(x_1$, $l_a)$, $(x_3, l_b)\}$, and \makebox{$R_6=\{(x_2, l_a), (x_3, l_c)\}$}. 
They correspond to 4 removal strategies:
\makebox{$\mathcal{S}_1=\{(l_a: 1)\}$}, \makebox{$\mathcal{S}_2=\{(l_c: 1)\}$}, $\mathcal{S}_3=\{(l_a:1)$, $(l_c:1)\}$, and \makebox{$\mathcal{S}_4= \{(l_a: 2)\}$}. 
As the number of elements in $T$ increases, the size gap between the concrete and abstract domains increases drastically--- this is the reason why our method is efficient.
%

\subsection{The \textsc{QuickCertify} Subroutine}

In this subroutine, we check a series of \emph{sufficient conditions} under which the prediction result for test input $x$ is guaranteed to be robust.  These sufficient conditions are designed to avoid the most expensive step of the KNN algorithm, which is the learning phase that relies on $p$-fold cross validations to compute the optimal $K$ parameter.

Since the optimal $K$ parameter is chosen from a set of candidate values, where $p$-fold cross validations are used to identify the value that minimizes prediction error, skipping the learning phase means we must directly analyze the behavior of the KNN prediction phase for all candidate $K$ values.  That is, assuming any of the candidate $K$ value may be the optimal one, we prove that the prediction result remains the same no matter which candidate $K$ value is used as the $K$ parameter.

Algorithm~\ref{alg:quick} shows the procedure, which takes the training set $T$, poisoning threshold $n$, and test input $x$ as input, and returns either \texttt{True} or \texttt{False} as output.  Here, \texttt{True} means the result is $n$-poisoning robust, and \texttt{False} means the result is unknown. 
For each candidate $K$ value,  $y=Freq(\mathcal{E}(T_x^{K}))$ is the most frequent label of the $K$-nearest neighbors of $x$.

\begin{algorithm}[t]
\caption{Subroutine $\textsc{QuickCertify} (T, n, x)$.}
\label{alg:quick}
{\footnotesize
$LabelSet \gets \{\}$\\
\For {each candidate $K$ value}  {
	Let $y = Freq(\mathcal{E}(T_x^{K}))$ and add $y$ into $LabelSet$;\\
		\If {$y \neq Freq(\mathcal{E}(T_x^{K + n}) \setminus \{(y : n)\})$} 
	{ \Return \texttt{False}}
		\If {$|LabelSet| > 1$}
	{ \Return \texttt{False}}
}
\Return \texttt{True}
}
\end{algorithm}

Recall that, in Section~\ref{sec:motivation}, we have explained the two ways in which poisoned data in $T$ may affect the prediction result.   The first one is called \emph{direct influence}: without changing the $K$ value, the poisoned data may affect the $K$-nearest neighbors of $x$ and thus their most frequent label. The second one is called \emph{indirect influence}: by changing the $K$ value, the poisoned data may affect how many neighbors to consider.
Inside the \textsc{QuickCertify} subroutine, we check for sufficient conditions under which none of the above two types of influence is possible.

The check for \emph{direct influence} is implemented in Line~4. Here, $T_x^{K+n}$ consists of the $(K+n)$ nearest neighbors of $x$, and $\mathcal{E}(T_x^{K+n})$ is the label counter.  Therefore, $\mathcal{E}(T_x^{K+n}) \setminus \{(y:n)\}$ means removing $n$ data elements labeled $y$.   $Freq(\mathcal{E}(T_x^{K+n}) \setminus \{(y:n)\})$ represents the most frequent label after the removal.
If it is possible for this removal strategy to change the most frequent label, then we conservatively assume that the prediction result \emph{may not} be robust.

The check for \emph{indirect influence} is implemented in Line~7.  Here, $LabelSet$ stores all of the most frequent labels for different candidate $K$ values.  If the most frequent labels for any two candidate $K$ values differ, i.e., $|LabelSet|>1$, we conservatively assume the prediction result \emph{may not} be robust.

On the other hand, if the prediction result remains the same during both checks, we can safely assume that the prediction result is $n$-poisoning robust.

\subsection{Two Examples}

We illustrate  Algorithm~\ref{alg:quick} using two examples.

\begin{figure}
 \centering
	\begin{subfigure}[b]{0.22\textwidth}
	 \centering
	\includegraphics[width=.9\textwidth]{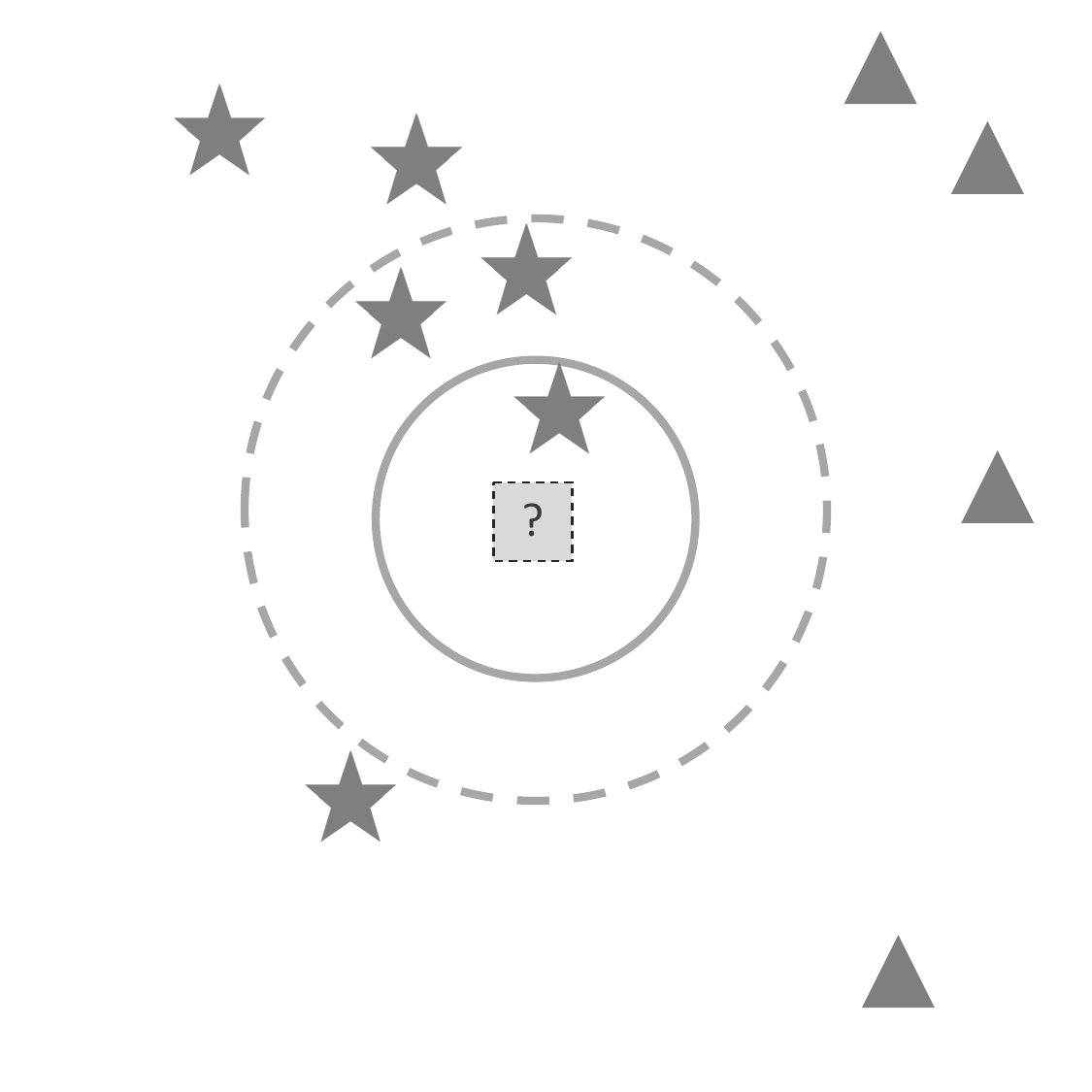}
	\caption{\scriptsize For $K=1$, $Freq(\mathcal{E}(T_x^{1})) = star$, and $Freq(\mathcal{E}(T_x^{1+n})\setminus\{star: n\}) = star$}
	\end{subfigure}
\hfill
	\begin{subfigure}[b]{0.22\textwidth}
	 \centering
	\includegraphics[width=.9\textwidth]{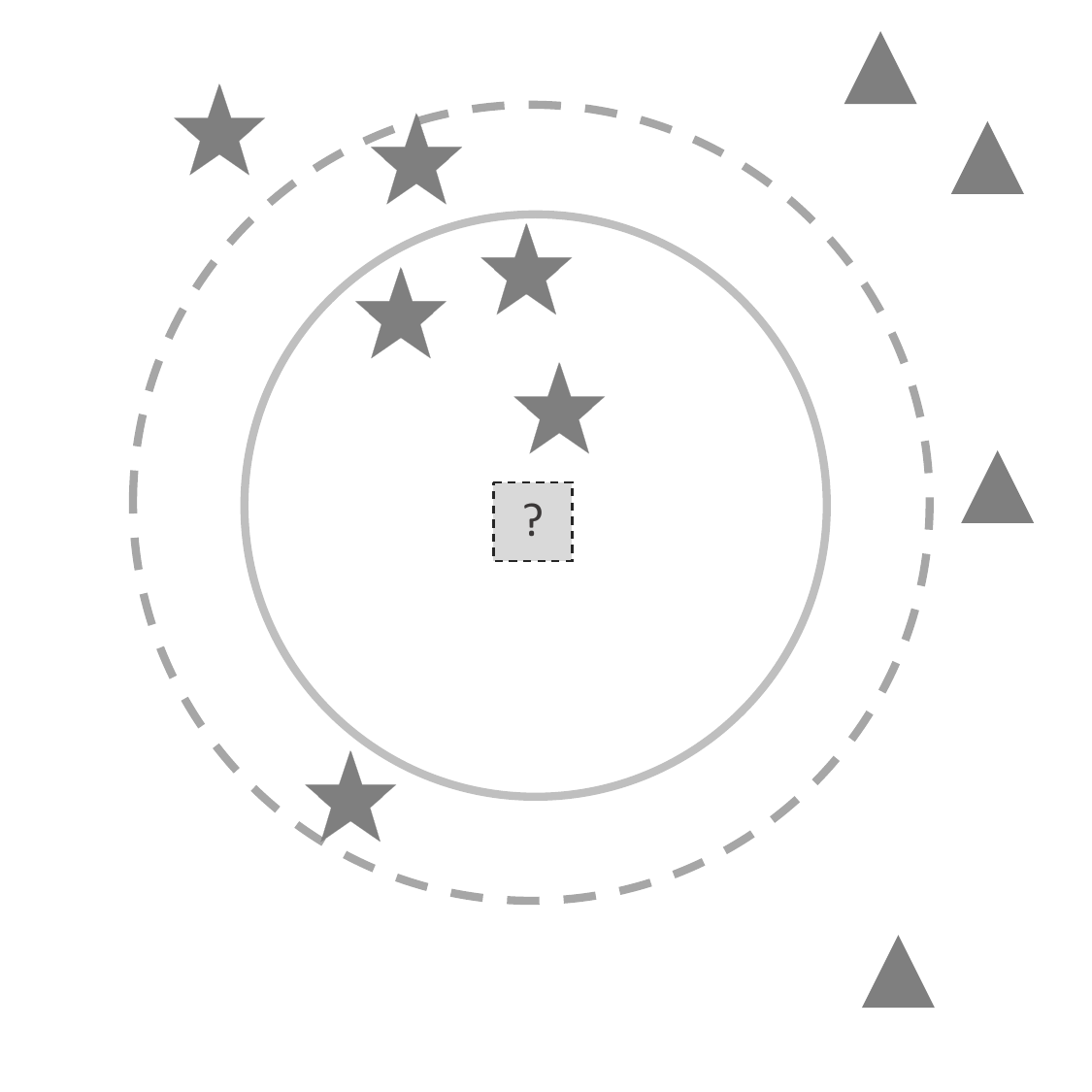}
	\caption{\scriptsize For $K=3$, $Freq(\mathcal{E}(T_x^{3})) = star$, and $Freq(\mathcal{E}(T_x^{3+n})\setminus\{star: n\}) = star$}
	\end{subfigure}
\caption{\emph{Robust} example for  \textsc{QuickCertify}, where the poisoning number is $n=2$, and candidate $K$ values are \{1, 3\}.}
\label{fig:case1}
\end{figure}

Figure \ref{fig:case1} shows an example where robustness can be proved by \textsc{QuickCertify}.   For simplicity, we assume the only two candidate values for the parameter $K$ are $K=1$ and $K=3$.  
When $K = 1$, as shown in Figure \ref{fig:case1} (a), \emph{star} is the most frequent label of the $x$'s neighbors, denoted $\mathcal{E}(T_x^1) = \{(star  : 1)\}$, and inside Algorithm~\ref{alg:quick}, we have $LabelSet = \{star \}$. 
The extreme case is represented by $\mathcal{E}(T_x^{1+2}) \setminus \{(star : 2)\} =  \{(star  : 1)\}$, which means $x$ is still classified as \emph{star} after applying this aggressive removal strategy.   

When $K = 3$, as shown in Figure \ref{fig:case1} (b), \emph{star} is also the most frequent label in $\mathcal{E}(T_x^3) = \{star : 3\}$ and thus $LabelSet =\{star\}$. 
The extreme case is represented by $\mathcal{E}(T_x^{3+2}) \setminus \{star : 2\} =  \{star : 3\}$, which means $x$ is still classified as \emph{star} after applying this removal strategy.
In this example $n=2$, thus $x$ is proved to be robust against 2-poisoning attacks.

\begin{figure}
 \centering
	\begin{subfigure}[b]{0.23\textwidth}
	 \centering
	\includegraphics[width=.9\textwidth]{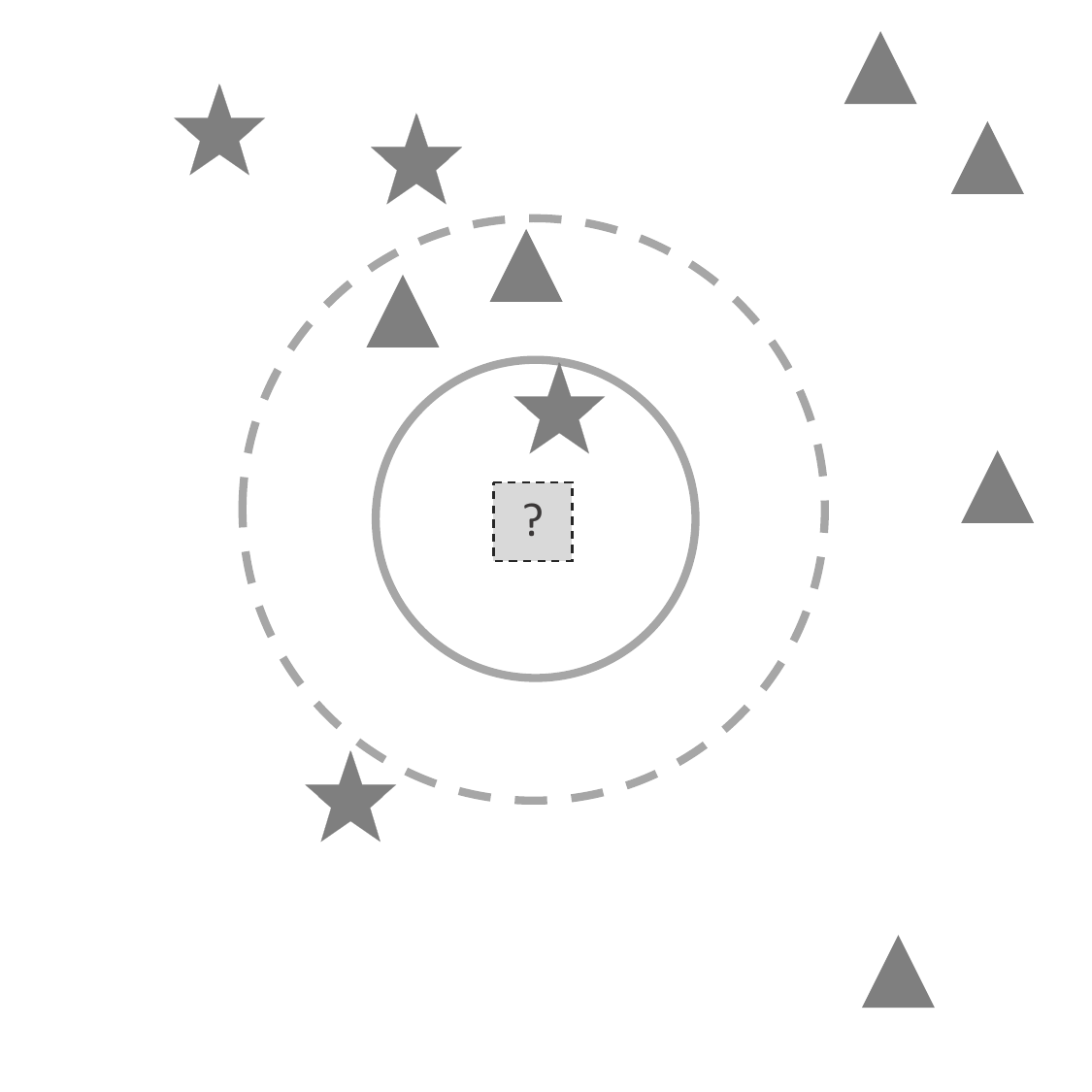}
	\caption{\scriptsize For $K=1$, $Freq(\mathcal{E}(T_x^{1})) = star$, and $Freq(\mathcal{E}(T_x^{1+n})\setminus\{star: n\}) = triangle$}
	\end{subfigure}
\hfill
	\begin{subfigure}[b]{0.24\textwidth}
	 \centering
	\includegraphics[width=.87\textwidth]{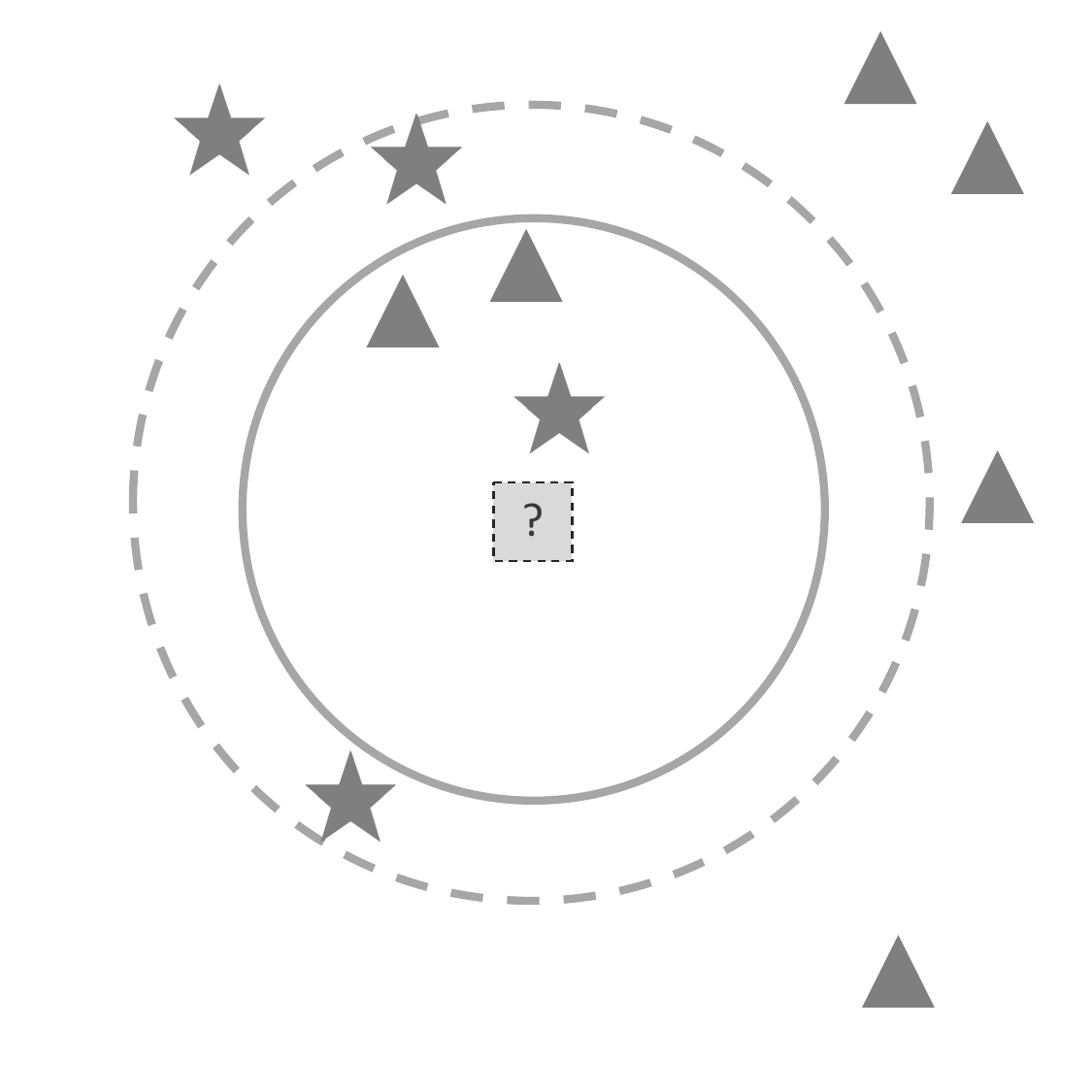}
	\caption{\scriptsize For $K=3$, $Freq(\mathcal{E}(T_x^{3})) = triangle$, and $Freq(\mathcal{E}(T_x^{3+n})\setminus\{triangle: n\}) = star$}
	\end{subfigure}
\caption{\emph{Unknown} example for \textsc{QuickCertify}, where the poisoning number is $n=2$ and the only two candidate values are $K$=1 and $K$=3.}
\label{fig:case2}
\end{figure}

Figure \ref{fig:case2} shows an example where the robustness cannot be proved by \textsc{QuickCertify}.  
%
%
When $K = 1$, as shown in Figure \ref{fig:case2} (a), \emph{star} is the most frequent label in $\mathcal{E}(T_x^1) = \{(star : 1)\}$ and $LabelSet = \{star\}$. 
The extreme case is $\mathcal{E}(T_x^{1+2}) \setminus \{(star: 2)\} =  \{triangle : 2\}$, which means $x$ is classified  as \emph{triangle}.  Thus, \textsc{QuickCertify} returns \texttt{False} in Line~5.

\subsection{Correctness and Efficiency}

The following theorem states that our method is sound in proving $n$-poisoning robustness. 

\begin{theorem}
\label{thm:quick}
If \textsc{QuickCertify}$(T,n,x)$ returns \texttt{True}, the KNN's prediction result for $x$ is guaranteed to be $n$-poisoning robust. 
\end{theorem}

Due to space limit, we omit the full proof. 
Instead, we explain the intuition behind Line~4 of the algorithm. 
First, we note that the prediction label $Freq(\mathcal{E}(T_x^{'K}))$ from any  \makebox{$T' \subset \Delta_n(T)$} can correspond to a $Freq(\mathcal{E}(D))$ where D is obtained by removing \makebox{$i~(\leq n)$} elements from $T_x^{K + n}$. Thus, we only need to pay attention to the $(K + n)$ nearest neighbors of $x$; other elements which are further away from $x$ can be safely ignored (cf.~\cite{li2022proving,jia2020certified}).
Next, to maximize the chance of changing the most frequent label from $y$ to another label, we want to remove as many $y$-labeled elements  as possible from $x$'s neighbors.  
Thus, the most aggressive removal case is captured by \makebox{$\mathcal{E}(T_x^{K + n}) \setminus \{(y : n)\}$}. If the most frequent label remains unchanged even in this case, it is guaranteed unchanged.

Next, we explain why \textsc{QuickCertify} is fast.  There are three reasons.  First, it completely avoids the computationally expensive $p$-fold cross validations.  Second, it considers only the $K+n$ nearest neighbors of $x$.  Third, it focuses on analyzing the label counts, which are in the (small) abstract domain, as opposed to the removal sets, which are in the (large) concrete domain.

For these reasons, the execution time of this subroutine is often negligible (e.g., less than 1 second) even for large datasets.  At the same time,  our experimental evaluation shows that it can prove robustness for a surprisingly large number of test inputs.

\textcolor{black}{
To summarize, mapping a potentially large set of concrete sets to their corresponding label multiset (label counts) is an over-approximated abstraction, since the prediction result for a test input $x$ is determined by the label counts of $x$'s nearest neighbors. This over-approximated abstraction allows \textsc{QuickCertify} to efficiently analyze the impact of the maximal allowable change in the label counts.
}

\section{Reducing the Search Space}
\label{sec:constraint}

In this section, we present the subroutine \textsc{GenPromisingSubsets}, which narrows down the search space by removing \emph{obviously non-violating} subsets from $\Delta_n(T)$ and returns the remaining ones, denoted by the set $\nabla_n^x(T)$ in Algorithm~\ref{alg:our_method}.

\subsection{Minimal Violating Removal in Neighbors}

We filter the \emph{obviously non-violating} subsets by computing some common property for each candidate $K$ value such that it must be part of every \emph{violating} removal set. 

We observe that any violating removal set for a specific candidate $K$ value must ensure that, for test input $x$, its new $K$ nearest neighbors after removal have a most frequent label $y'$ that is different from the default label $y$. Our method computes the minimal number of removed elements in $x's$ neighborhood to achieve this, let us call it minimal violating removal, denote $min\_rmv$. With this number, we know the every \emph{violating} removal set must have at least $min\_rmv$ elements from $x$'s neighbors $T_x^{K+n}$.

The test input $x$’s new nearest neighbors after removal is represented as $T_x^{K+i} \setminus \{i$ elements from $T_x^{K+i}\}$, where $i =1,2,...n$. To compute the minimal violating removal, rather than checking each possible value of $i$ from $1$ to $n$, we need a more efficient method, e.g., binary search with $O(log~{n})$. To use binary search, we need to prove the monotonicity of violating removals, defined below.

\begin{theorem}[Monotonicity]
\label{the:mono}
If there is some $i$ allowing $T_x^{K+i} \setminus \{i$ elements from $T_x^{K+i}\}$ to have a different most-frequent label $y’$, then any larger value $j > i$ will also allow $T_x^{K+j} \setminus \{j$ elements from $T_x^{K+j}\}$ to have a different most-frequent label $y’$. Conversely, if $i$ does not allow it, then any smaller value $j < i$ does not allow it either.
\end{theorem}

\begin{proof}
If there is some $i$ allowing $T_x^{K+i} \setminus \{i$ elements from $T_x^{K+i}\}$ to have a different most-frequent label $y’$, there exists $S \subset T_x^{K+i}$ such that $|S| = i$ and $Freq(T_x^{K+i}\setminus S) = y'$. For any $j > i$ and $T_x^{K+j}$, we can always construct $S' = S \cup (T_x^{K+j}\setminus T_x^{K+i})$, which satisfies $S' \subset T_x^{K+j}, |S'| = j$ and $Freq(T_x^{K+j}\setminus S') = y'$. 
The reverse can be proved similarly. 
%
\end{proof}

Lines~ 2-11 in Algorithm \ref{alg:abs-analyze}  show the process of finding the minimal violating removal using binary search. Assume the possible range is $0\sim n+1$ (line~2), the binary search divides the range in half (line~4) and checks the middle value (line~5). To check whether a removal $mid$ can result in a different label $y'\neq y$, the most possible operation is to remove $mid$ elements with $y$ label. It $mid$ works, according to Theorem \ref{the:mono}, we know the minimal removal is in the range $start\sim mid$ (line~6); otherwise it is in the range $mid+1\sim end$ (line~8). The binary search stops when $start$ equals $end$, and this will the minimal violating removal.

\begin{algorithm}[t]
\caption{$\textsc{GenPromisingSubsets}(T, n, x, y)$.}
\label{alg:abs-analyze}
{\footnotesize
\For {each candidate $K$ value} {
	$start = 0; end = n + 1;$\\ 
	\While {start < end} {
 		$mid = (start + end) / 2;$\\
		\eIf {$y \neq Freq(\mathcal{E}(T_x^{K+m}) \setminus \{(y : m)\})$} {
			$end = mid;$
		}{ $start = mid + 1;$ }
	}
	$min\_rmv = start;$\\
  	\If {$min\_rmv \leq n$} {
           \For{each $R_1\subseteq T_x^{K+n}$ s.t. $|R_1| \geq min\_rmv$}{
    	      \For {each $R_2 \subseteq (T \setminus T_x^{K+n})$ and $|R_2| \leq n - |R_1|$} {
		  $R = R_1 \cup R_2$;\\
		  Add $(T\setminus R)$ to $\nabla_n^x(T);$
	      }
           }
        }
}
}
\end{algorithm}

Since $n$ is the maximal allowed removal, when $min\_rmv > n$, it is impossible for the most frequent label to change from $y$ to $y'$.

\subsection{An Illustrative Example}

Here we give an example of the binary search in Algorithm \ref{alg:abs-analyze}. Assume in the original training set $T$, for the test input $x$, the optimal $K$ is $K = 1$ and the default label is $y = star$. 

\begin{example}
Assume $n = 5$, $T_x^3=\{star *2, triangle * 1\}$, $T_x^4=\{star *2, triangle * 2\}$, and $T_x^5=\{star *3, triangle * 2\}$. For the candidate $K = 2$, we show how to compute the minimal violating removal in $x's$ neighbors. 

At first, $start=0$ and $end=6$, which means the possible value range of minimal removal is $0\sim6$. Our method first checks $mid = 3$, since $T_x^{2+3} \setminus \{(star : 3)\}$ results in the most-frequent label $triangle$, our method can cut the possible range by half to $0\sim3$. Next, we check $mid=1$, and reduce the range to $0\sim1$. Finally, we check $mid = 0$, which does not work, so the range becomes $1\sim1$, and we return 1 as the minimal violating removal in  $x$'s neighbors.
\end{example}

Since binary search reduces the range by half at each step, it is efficient.  For example, when $n$=180 for MNIST, binary search needs only 8 checks to compute the result,  whereas going through each value in the range requires 180 checks.  In other words, the speedup is more than 20X.

\subsection{The Reduced Search Space}

Based on the minimal violating removal, $min\_rmv$, we compute the reduced set $\nabla_n^x(T)$ as shown in Lines~12-20 of Algorithm~\ref{alg:abs-analyze}.

Here, each removal set $R$ is the union of two sets, $R_1$ and $R_2$, where $R_1$ is a removal set that contains at least $min\_rmv$ elements from $x$'s neighborhood  $T_x^{K+n}$, and $R_2\subseteq (T\setminus T_x^{K+n})$ is a subset of the left-over data elements.

Our experiments show that, in practice, the reduced set $\nabla_n^x(T)$ is often significantly smaller than the original set $\Delta_n(T)$. 
A special case is when $min\_rmv = 0$, for which $\nabla_n^x(T)$ is the same as $\Delta_n(T)$, meaning the search space is not reduced.  
However, this special case is rare and, during our experimental evaluation, it never occurred.

\section{Incremental Computation}
 \label{sec:difflearn}

In this section, we present our method for speeding up an expensive step of the KNN algorithm, the $p$-fold cross validations inside \textsc{KNN\_learn}.
We achieve this speedup by splitting \textsc{KNN\_learn} into two subroutines:  \textsc{KNN\_learn\_init}, which is applied only once to the original training set $T$, and \textsc{KNN\_learn\_update}, which is applied to each individual removal set $R=(T\setminus T')$, where $T'\in\nabla_n^x(T)$.

\subsection{The Intuition}

First, we explain why the standard \textsc{KNN\_learn} is computationally expensive. 
This is because, for each candidate value of parameter $K$, denoted $K_i$,  the standard $p$-fold cross validation~\cite{mclachlan2005analyzing} must be used to compute the classification error.   
Algorithm~\ref{alg:knn-learn2} (excluding Lines~15-16) shows the computation.

First, the training set $T$ is partitioned into $p$ groups, denoted $\{G_1, G_2, ..., G_p\}$.  Then, the set of misclassification samples in each group $G_j$ is computed, denoted  $errSet_{G_j}^{K_i}$.  Next, the error is averaged over all groups, which results in $error^{K_i}$.  Finally, the $K_i$ value with the smallest classification error is chosen as the optimal $K$ value.

The computation is expensive because $error_{G_j}^{K_i}$, for each $K_i$, requires exactly $|G_j|$ calls to the standard \textsc{KNN\_predict}$(T\setminus G_j, K_i, x)$, one per data element $x\in G_j$, while treating the set $D=(T\setminus G_j)$ as the training set.

Our intuition for speeding up this computation is as follows.
Given the original training set $T$, and a subset $T'\in\nabla_n^x(T)$, the corresponding removal set $R=(T\setminus T')$ can capture
the difference between these two sets, and thus capture the difference of their $error^{K_i}$.
Since $K_i$ is fixed when computing $error^{K_i}$, we only need to consider the \emph{direct influence} (i.e., neighbors change) brought by removal set $R$.
In practice, the removal set is often small, which means the vast majority of data elements in the $p$-fold partition of $T'$, denoted $\{G'_1,\ldots,G'_p\}$, are the same as data elements in the $p$-fold partition of $T$, denoted $\{G_1,\dots,G_p\}$.
Thus, for most elements, their neighbors are almost the same.
Instead of computing the error sets ($errSet_{G'_j}^{K_i}$) from scratch for every single $G'_j$, we can use the error sets ($errSet_{G_j}^{K_i}$) for $G_j$ as the starting point, and only compute the change brought by removal set $R$, leveraging the intermediate computation results stored in $Error$.


\subsection{The Algorithm}

Our incremental computation has two steps.  As shown in Algorithm~\ref{alg:our_method}, we apply \textsc{KNN\_learn\_init} once to the set $T$, and then apply \textsc{KNN\_learn\_update} to each removal set $R=(T\setminus T')$.

Our new subroutine \textsc{KNN\_learn\_init} is shown in Algorithm~\ref{alg:knn-learn2}. It differs from the standard \textsc{KNN\_learn} only in Lines~15-16, where it stores the intermediate computation results in \emph{Error}.  The first component in \emph{Error} is the set of $p$ groups in $T$.  The second component contains, for each $K_i$,  the misclassified elements in $G_j$.

\begin{algorithm}[t]
\caption{Subroutine $\textsc{KNN\_learn\_init}(T)$.}
\label{alg:knn-learn2}
{\footnotesize
Partition the training set $T$ into $p$ groups $\{G_1, G_2, ..., G_p\}$\\
\For {each candidate $K_i$ value} {
   \For {each group $G_j$} {
      $errSet_{G_j}^{K_i} \gets \{\}$\\
      \For {each data element $(x,y) \in G_j$}{
         \If {$\textsc{KNN\_predict}(T \setminus G_{j}, K_i, x) \neq y$}	
              {Add $(x, y)$ to $errSet_{G_j}^{K_i}$;}     
      }
      $error_{G_j}^{K_i} = \left| errSet_{G_j}^{K_i} \right| ~/~ \left| G_j \right|$\\
   }
   $error^{K_i} = \frac{1}{p}\sum_{j = 1}^{p} error_{G_j}^{K_i}$\\
}
$K \gets \underset{K_i}{\textbf{argmin}}~~   error^{K_i}$\\
\textcolor{darkblue}{
$Error \gets \langle \{G_1, G_2, ..., G_p\}, \{(errSet_{G_1}^{K_i}, \dots,  errSet_{G_p}^{K_i} )\} \rangle$
}\\

\Return $\langle K, Error \rangle$
}
\end{algorithm}

\begin{algorithm}[t]
\caption{$\textsc{KNN\_learn\_update}(R,  Error)$.}
\label{alg:diff-learn}
{\footnotesize

Let $\{G_1, \ldots, G_p\}$ and $\{(errSet_{G_1}^{K_i}, \dots,  errSet_{G_p}^{K_i} )\}$ be groups and error sets stored in $Error$\\
Compute the new groups $\{G'_j \mid G'_j = G_j \setminus R \mbox{ where } j = 1, \ldots, p\}$\\
Compute the new training set  $T' = \bigcup_{j\in\{1,\ldots,p\}} ~ G'_j$\\
Compute the influenced set, $influSet$, using $R$ and $\{G_j\}$\\

\For {each candidate $K_i$ value}{
   	\For {each new group $G'_j$} {
\textcolor{darkblue}{
           $newSet^+ = newSet^- = \{\}$}\\
           \For {each data element $(x,y) \in (G'_j \textcolor{darkblue}{ \cap influSet} )$ }{  
              \If {$\textsc{KNN\_predict}(T \setminus G_{j}, K_i, x) = y$ and $\textsc{KNN\_predict}(T' \setminus G'_{j}, K_i, x) \neq y$}	
                 {Add $(x, y)$ to $newSet^+$;} 
              \If {$\textsc{KNN\_predict}(T \setminus G_{j}, K_i, x) \neq y$ and $\textsc{KNN\_predict}(T' \setminus G'_{j}, K_i, x) = y$}
                 {Add $(x, y)$ to $newSet^-$;} 
           }           
\textcolor{darkblue}{
      	   $errSet_{G'_j}^{K_i} = errSet_{G_j}^{K_i} \setminus R \setminus newSet^- \cup newSet^+$ 
}\\

  	   $error_{G'_j}^{K_i} = \left| errSet_{G'_j}^{K_i} \right| ~/~ \left| G'_j \right|$\\
	}
   $error^{K_i} = \frac{1}{p}\sum_{j = 1}^{p} error_{G'_j}^{K_i}$\\
}
$K \gets \underset{K_i}{\textbf{argmin}}~~   error^{K_i}$\\
\Return $K$
}
\end{algorithm}

Subroutine \textsc{KNN\_learn\_update} is shown in Algorithm \ref{alg:diff-learn}, which computes the new $errSet_{G'_j}^{K^i}$ based on the $errSet_{G_j}^{K_i}$ stored in $Error$.
First, it computes the new groups $G'_j$ by removing elements in $R$ from the old groups $G_j$. 
Then, it computes $influSet$, which is defined in the next paragraph. 
Finally, it modifies the old $errSet_{G_j}^{K_i}$ (in Line~16) based on three cases:  it removes the set $R$ (Case 1) and the set $newSet^-$ (Case 2), and adds the set $newSet^+$ (Case 3). 
Below are the detailed explanations of these three cases:
\begin{enumerate}
\item If $(x, y) \in G_j \setminus G'_j$ was misclassified by $(T\setminus G_j)$, but this element is no longer in $T'$, it should be removed. 
\item If $(x, y) \in G_j \cap G'_j$ was misclassified by $(T\setminus G_j)$, but this element is correctly classified by $T'\setminus G'_j$, it should be removed. 
\item If $(x, y) \in G_j \cap G'_j$ was correctly classified by $(T\setminus G_j)$, but is misclassified by $T'\setminus G'_j$, it should be added. 
\end{enumerate}
%
Case (1) can be regarded as an \emph{explicit change} brought by the removal set $R$, whereas Case (2) and Case (3) are \emph{implied change}s brought by $R$:  these changes are implied because, while the element $(x, y)$ is not inside $R$, it is classified differently after the elements in $R$ are removed from $T$.

Since the removal set is small, most data elements in $G_j$ will not be part of the explicit or implied changes.  To avoid redundantly invoking \textsc{KNN\_predict} on these data elements, we filter them out using the influenced set (Line~8).  Here, assume that $K_{max}=max(\{K_i\})$ is the maximal candidate value, and during cross-validation, when $G_j$ is treated as the test set, $D=(T\setminus G_j)$ is the corresponding training set. 

{\footnotesize
\[\begin{array}{ll}
  influSet = \{~ (x, y) \in G_j \mid 
                            &
                            (x, y) \not\in R,\\
                            & 
                            D_x^{K_{max}} \cap R \neq \emptyset, 
                            \mbox{ and } \\
                            &\textsc{QuickCertify}(D,n,x) = \texttt{False}
              \}
\end{array}\]}%

\vspace{1ex}
\noindent
In other words, every element $(x, y)$ inside $influSet$ must satisfy three conditions: 
(1) the element is not in $R$; 
(2) at least one of its neighbors in $D_x^{K_{max}}$ is in $R$; and
(3) the element may be misclassified when at most $n$ neighbors are removed. 
Recall that the subroutine used in the last condition has been explained in Algorithm~\ref{alg:quick}.

\begin{table*}[h]
\centering
\caption{Comparing the accuracy of our method with the baseline (ground truth)and two existing methods (which cannot falsify) on the smaller datasets, for which the ground truth can be obtained by the baseline enumerative method (Algorithm~\ref{alg:baseline}).}
\label{tbl:small}
\scalebox{0.7}{
\begin{tabular}{|l|c|c|c|c|r|c|c|c|r|c|c|c|r|c|c|c|r|}
  \hline 
\multicolumn{2}{|c|}{Benchmark}  & \multicolumn{4}{c|}{Baseline} 
                                  & \multicolumn{4}{c|}{Jia et al.~\cite{jia2020certified}}
                                  & \multicolumn{4}{c|}{Li et al.~\cite{li2022proving}}   
                                  & \multicolumn{4}{c|}{Our Method}  
\\\cline{1-18}
%
   dataset  &  test data     & certified  & falsified  & \textbf{unknown}   & time
                            & certified  & falsified  & \textbf{unknown}   & time
                            & certified  & falsified  & \textbf{unknown}   & time
                            & certified  & falsified  & \textbf{unknown}   & time

\\
            &  \#         & \#  & \#  & \#   & (s)
                            & \#  & \#  & \#   & (s)
                            & \#  & \#  & \#   & (s)
                            & \#  & \#  & \#   & (s)
\\\hline\hline

Iris ($n$=1)     &  15    &    15  & 0  &   \textbf{0}   &   49    
                          &    0   & 0  &   \textbf{15}  &    1    
                           &    14   & 0  &   \textbf{1}  &    1    
                          &    15  & 0  &   \textbf{0}   &    1        
\\\hline

Iris ($n$=2)     &  15    &    14  &  1 &   \textbf{0}   &   3,086     
                          &    0   & 0  &   \textbf{15}  &    1  
                          &    13   & 0  &   \textbf{2}  &    1      
                          &    14  & 1  &   \textbf{0}   &    5        
 \\\hline

Iris ($n$=3)     &  15    &    0  & 1  &   \textbf{14}   &   6,721     
                          &    0   & 0  &  \textbf{15}  &    1 
                          &    11   & 0  &  \textbf{4}  &    1       
                          &    13  & 1  &  \textbf{1}   &    120       
\\\hline

Digits ($n$=1)   & 180    &    0   & 1 &   \textbf{179} &  7,168   %
                          &    170   & 0  &  \textbf{10}   &    1        %
                          &    172   & 0  &  \textbf{8}  &    1    
                          &    179 & 1  &  \textbf{0}   &    3    %
\\\hline
          \end{tabular}
}
\end{table*}

\section{Experiments}
\label{sec:expr}

We have implemented our method using Python and the popular machine learning toolkit \texttt{scikit-learn 0.24.2}, together with the baseline method in Algorithm~\ref{alg:baseline}, and the two existing methods of Jia et al.~\cite{jia2020certified} and Li et al.~\cite{li2022proving}.  
%
For experimental comparison, we used six popular supervised learning datasets as benchmarks. 
There are two relatively small datasets, Iris \cite{fisher1936use} and Digits~\cite{gates1972reduced}. Iris has 135 training and 15 test elements with 3 classes and 4-D features. Digits has 1,617 training and 180 test elements with 10 classes and 64-D features.
Since the baseline approach (Algorithm~\ref{alg:baseline}) can finish on these small datasets and thus obtain the ground truth (i.e., whether prediction is truly robust), these small datasets are useful in evaluating the accuracy of our method.

The other four benchmarks are larger datasets, including HAR (human activity recognition using smartphones)~\cite{anguita2013public}, which has  9,784 training and 515  test elements with 6  classes and 561-D features, 
Letter (letter recognition)~\cite{frey1991letter}, which has  18,999  training and 1,000  test elements with 26 classes and 16-D features, 
MNIST (hand-written digit recognition)~\cite{lecun1998gradient}, which has 60,000  training and 10,000 test elements with 10 classes and 36-D features,  
and CIFAR10 (colored image classification)~\cite{krizhevsky2009learning}, which has 50,000 training and 10,000  test elements with 10 classes and 288-D features.  
Since none of these datasets can be handled by the baseline approach, they are used primarily to evaluate the efficiency of our method.
%

\subsection{Evaluation Criteria}

Our experiments aimed to answer the following three research questions: 
\begin{itemize}
\item [RQ1] 
Is our method accurate enough for deciding (certifying or falsifying) $n$-poisoning robustness for most of the test cases? 
%
%
\item [RQ2]
Is our method efficient enough for handling all of the datasets used in the experiments?  
%
%
\item [RQ3]
How often can prediction be successfully certified or falsified by our method, and how is the result affected by the poisoning threshold $n$?
%
\end{itemize}

We used the state-of-the-art implementation of KNN in our experiments, with 10-fold cross validation and candidate $K$ values in the range $1 \sim \frac{1}{10}|T|$.
The set $T$ is obtained by inserting up-to-$n$ malicious samples to the datasets. We first generate a random number $n' \leq n$, and then insert exactly $n'$ mutations of randomly picked input features and output labels of the original samples.

We ran all four methods on all datasets.  
%
For the slow baseline, we set the time limit to 7200 seconds per test input.  
For the other methods, we set the time limit to 1800 seconds per test input.  
%
Our experiments were conducted (single threaded) on a CloudLab~\cite{Duplyakin+:ATC19} c6252-25g node with 16-core AMD 7302P at 3 GHz CPU and 128GB EEC Memory (8 x16 GB 3200MT/s RDIMMs).

\subsection{Results on the Smaller Datasets}

To answer RQ1, we compared the result of our method with the ground truth obtained by the baseline enumerative method on the two smallest datasets.

Table~\ref{tbl:small} shows the experimental results. Columns~1-2 show the name of the dataset, the poisoning threshold $n$, and the number of test data.  Columns~3-6 show the result of the baseline method, including the number of test data that are certified, falsified, and unknown, respectively, and the average  time per test input. 
The remaining columns compare the results of the two existing methods and our method. 
Since the goal is to compare our method with the ground truth (obtained by the baseline method), we must choose small $n$ values to ensure that the baseline method does not time out.

On Iris ($n=2$), the baseline method was able to certify 14/15 of the test data and falsify 1/15. However, it was slow: the average time was 3,086 seconds per test input.
In contrast, the method by Jia et al.~\cite{jia2020certified} was much faster, albeit with low accuracy. It took 1 second per test input, but failed to certify any of the test data.   
The method by Li et al.~\cite{li2022proving} certified 11/15 of the test data but left 4/15 as unknown.
Our method certified 14/15 of the test data and falsified the remaining 1/15, and thus is as accurate as the ground truth; the average time is 5 seconds per test input.

While the slow baseline method was able to handle Iris, it did not scale well.  With a slightly larger dataset or larger poisoning threshold, it would run out of time.  
On Digits ($n$=1), the baseline method falsified only 1/180 of the test data and returned the remaining 179/180 as unknown.
In contrast, our method successfully certified or falsified all of the 180 test data.

\begin{table}
\centering
\caption{Comparing the accuracy and efficiency of our method with  existing methods on all datasets, with large poisoning thresholds;
the percentages of \emph{certified} and \emph{falsified} cases are reported in Section~\ref{sec:poison-impact} and shown in Figure~\ref{fig:certified}.
}
\label{tbl:our}
\scalebox{0.75}{
\begin{tabular}{|l|c|c|c|c|c|c|c|}
  \hline
   \multicolumn{2}{|c|}{Benchmark} & \multicolumn{2}{c|}{Jia et al.~\cite{jia2020certified}}  
    				   & \multicolumn{2}{c|}{Li et al.~\cite{li2022proving}} 
   				   & \multicolumn{2}{c|}{Our Method} 
\\\hline
    dataset  & poisoning &  unknown   &  time &  unknown  & time  & unknown  & time 
\\
             &  threshold  	&   \%        &   (s) &  \%       &  (s)  & \%       &  (s)
\\\hline\hline

Iris    & $n=$3 (2\%)      & 100\%      & 1	& 26.7\%	& 1	& 6.7\%	&  120   \\\hline
Digits  & $n=$16 (1\%)     & 100\%  	& 1	& 19.4\%	& 1	& 1.0\%	&  19    \\\hline
HAR     & $n=$97 (1\%)     & 100\%  	& 1	& 28.3\%	& 1	& 0.8\%	& 21 \\\hline
Letter  & $n=$190 (1\%)    & 100\% 	& 1	& 94.5\%	& 1	& 0.0\%	& 4\\\hline
MNIST   & $n=$180 (0.3\%)  & 38.1\% 	& 1  	& 25.0\%	& 1	& 2.0\% & 47 \\\hline
CIFAR10 & $n=$150 (0.3\%)  & 90.0\% 	& 1	& 64.0\%	& 1	& 0.0\%	& 558 \\\hline
\end{tabular}
}
\end{table}

\subsection{Results on All Datasets}

To answer RQ2, we compared our method with the two state-of-the-art methods~\cite{jia2020certified,li2022proving} on all datasets, using significantly larger poisoning thresholds. 
Since these benchmarks are well beyond the reach of the baseline method, we no longer have the ground truth. However, whenever our method returns \emph{Certified} or \emph{Falsified}, the results are guaranteed to be conclusive. Thus, the \emph{Unknown} cases are the only unresolved cases.
If the percentage of \emph{Unknown} cases is small, it means our method is accurate.

Table~\ref{tbl:our} shows the results, where Column~1 shows the name of the dataset, and Column~2 shows the poisoning threshold.   
For the smallest dataset,  we set $n$ to be 2\% of the size of $T$.  For medium datasets, we set it to be 1\%.  For large datasets, we set it to be 0.3\%.

Columns~3-6 show the percentage of test data left as \emph{unknown} by the two existing methods and the average time taken.  Recall that these methods can only certify, but not falsify, $n$-poisoning robustness. 

Columns~7-8 show the percentage of test data left as \emph{unknown} by our method.  While our method has a higher computational cost, it is also drastically more accurate than the two existing methods.

On HAR, for example, the existing methods left 100\% and 28.3\% of the test data as unknown when $n=97$.
Our method, on the other hand, left only 0.8\% of the test data as unknown.

On CIFAR10, which has 50,000 data elements with 288-D feature vectors, our method was able to resolve 100\% of the test cases when the poisoning threshold was as large as $n=150$. 
In contrast, the two existing methods resolved only 10.0\% and 36.0\%.  
In other words, they left 90.0\% and 64.0\% as unknown.

\subsection{Effectiveness of Our Method and Impact of the Poisoning Threshold}
\label{sec:poison-impact}

To answer RQ3, we studied the percentages of \emph{certified}, \emph{falsified}, and \emph{unknown} cases reported by our method, as well as how they are affected by the poisoning threshold $n$.

In addition to the percentage of \emph{unknown} cases shown in Table~\ref{tbl:our}, we show the percentages of \emph{certified} and \emph{falsified} cases reported by our method below. There is no need to report these percentages for the two existing methods, because they always have 0\% of \emph{falsified} cases.

\vspace{2ex}
\noindent
\scalebox{0.75}{
\begin{tabular}{l|c|c|c}
\hline
    dataset  & poisoning threshold  &  certified by our method  & falsified by our method \\
\hline
Iris    & $n=$3 (2\%)      & 86.6\%     &  6.7\%	\\\hline
Digits  & $n=$16 (1\%)     & 80.0\%  	& 19.0\%	\\\hline
HAR     & $n=$97 (1\%)     & 71.8\%     & 26.8\%	\\\hline
Letter  & $n=$190 (1\%)    & 5.6\% 	& 94.4\%	\\\hline
MNIST   & $n=$180 (0.3\%)  & 75.0\% 	& 23.0\%	\\\hline
CIFAR10 & $n=$150 (0.3\%)  & 36.0\% 	& 64.0\%	\\\hline
\end{tabular}
}
\vspace{2ex}

Figure~\ref{fig:certified} shows  how these percentages are affected by the poisoning threshold. 
Here, the $x$-axis shows $n/|T|$ in percentage, and the $y$-axis shows the percentages of \emph{falsified} in `$-$', \emph{unknown} in `$.$' and \emph{certified} in either `$|$' (quick certify) or `$/$' (slow certify).

\begin{figure}
     \centering
         \begin{subfigure}[b]{0.23\textwidth}
         \centering
         \includegraphics[width=\textwidth]{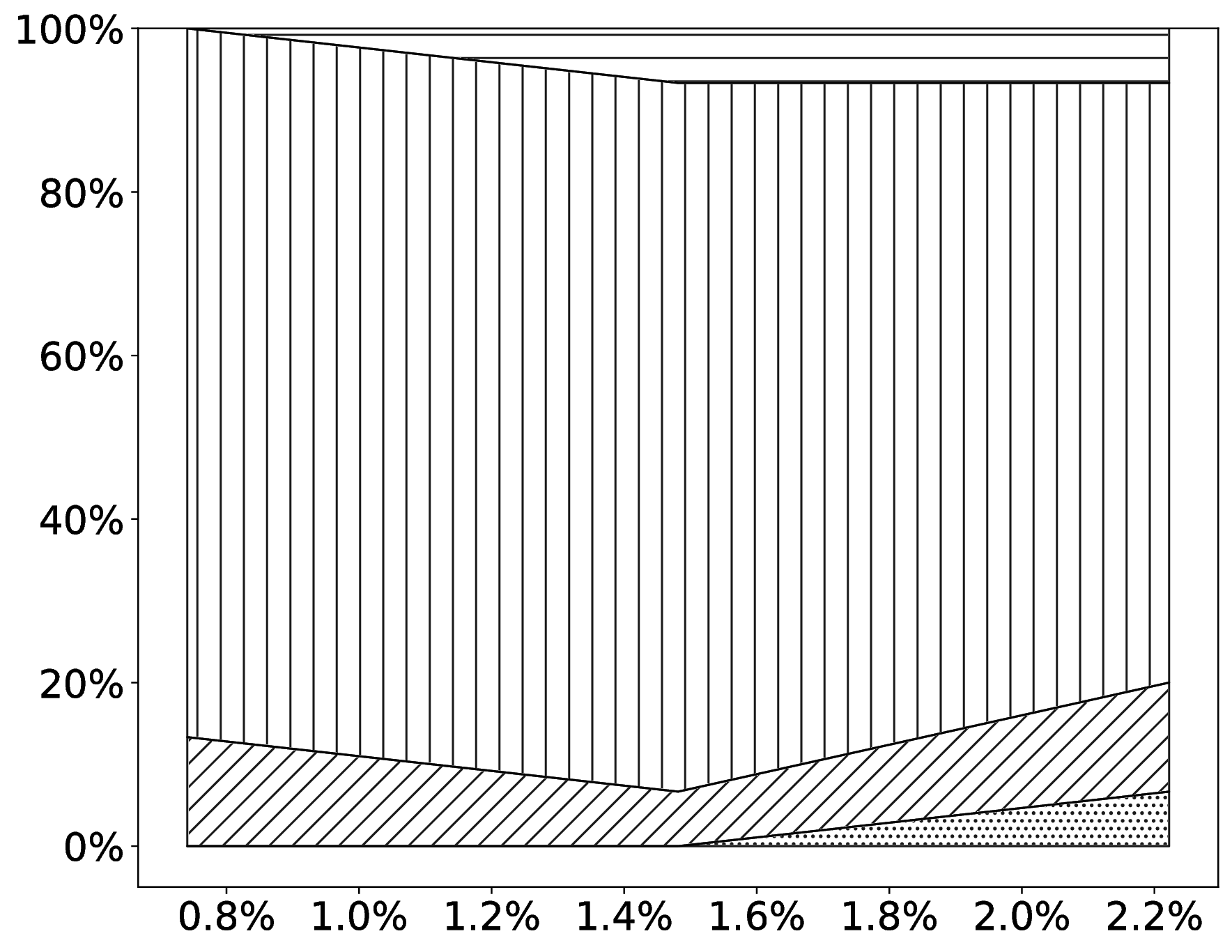}
         {\scriptsize (a) Iris}
         \end{subfigure}
      \hfill
     \begin{subfigure}[b]{0.23\textwidth}
         \centering
         \includegraphics[width=\textwidth]{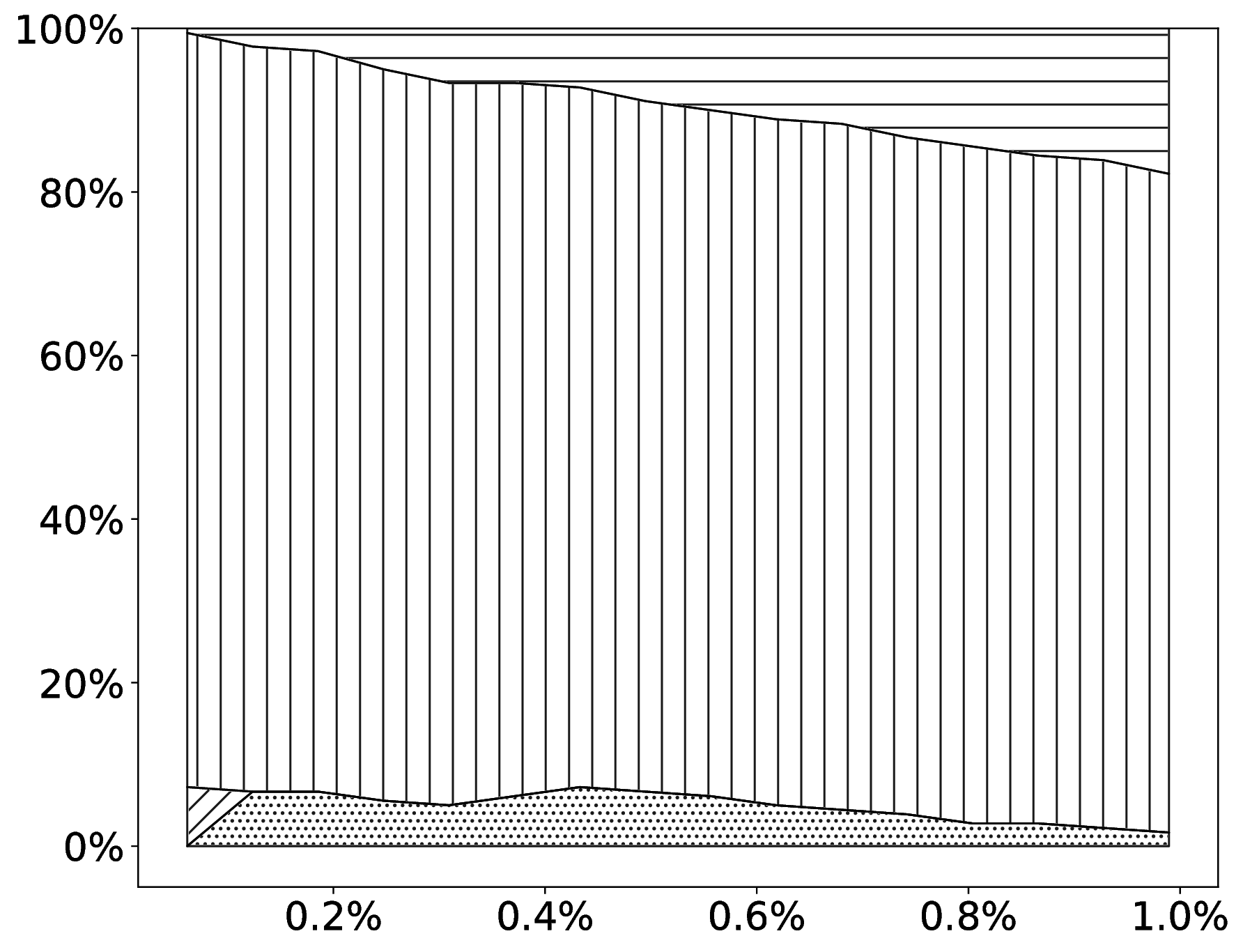}
         {\scriptsize (b) Digits}
            \end{subfigure}
     \hfill
         \begin{subfigure}[b]{0.23\textwidth}
         \centering
         \includegraphics[width=\textwidth]{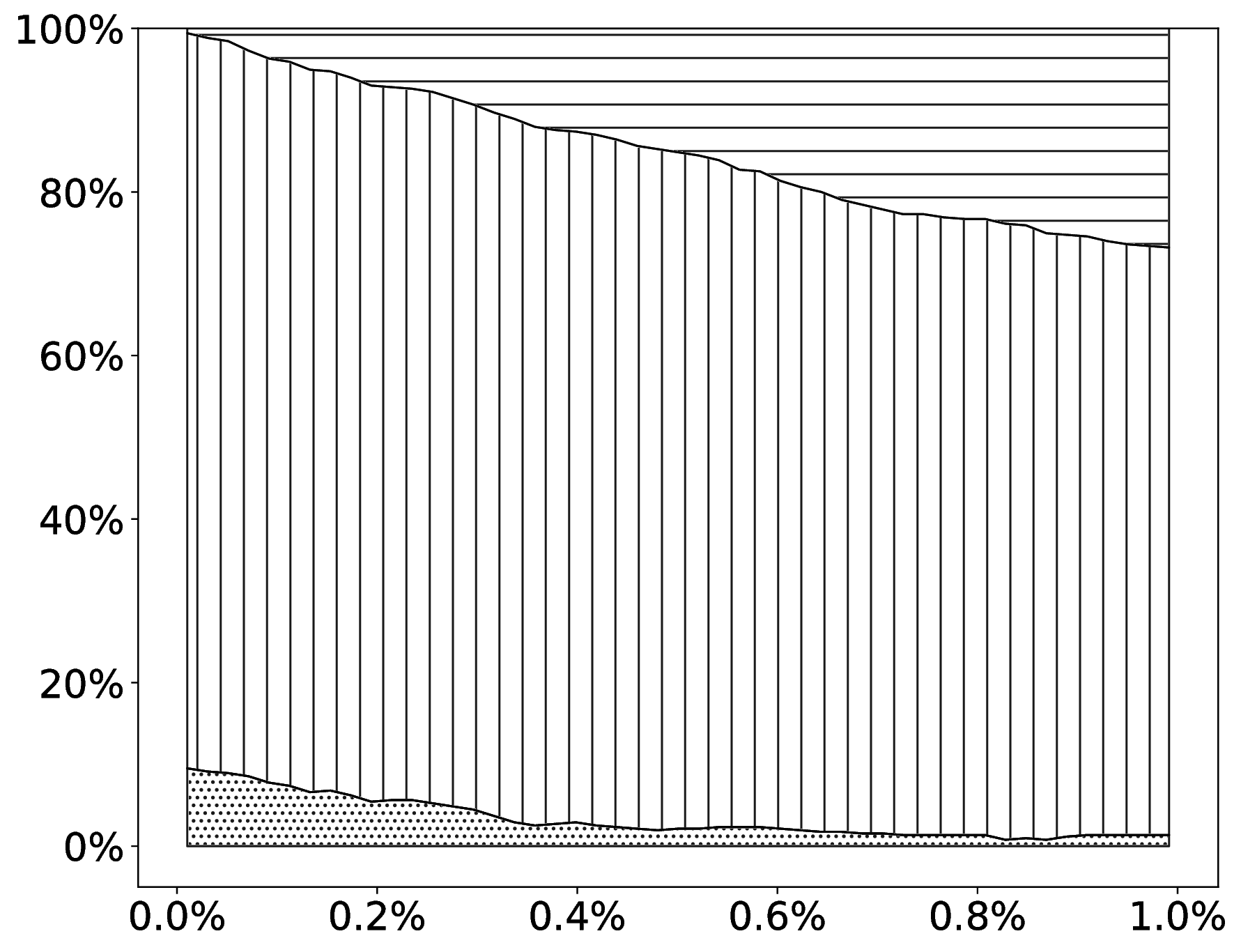}
         {\scriptsize (c) HAR}
          \end{subfigure}
     \hfill
         \begin{subfigure}[b]{0.23\textwidth}
         \centering
         \includegraphics[width=\textwidth]{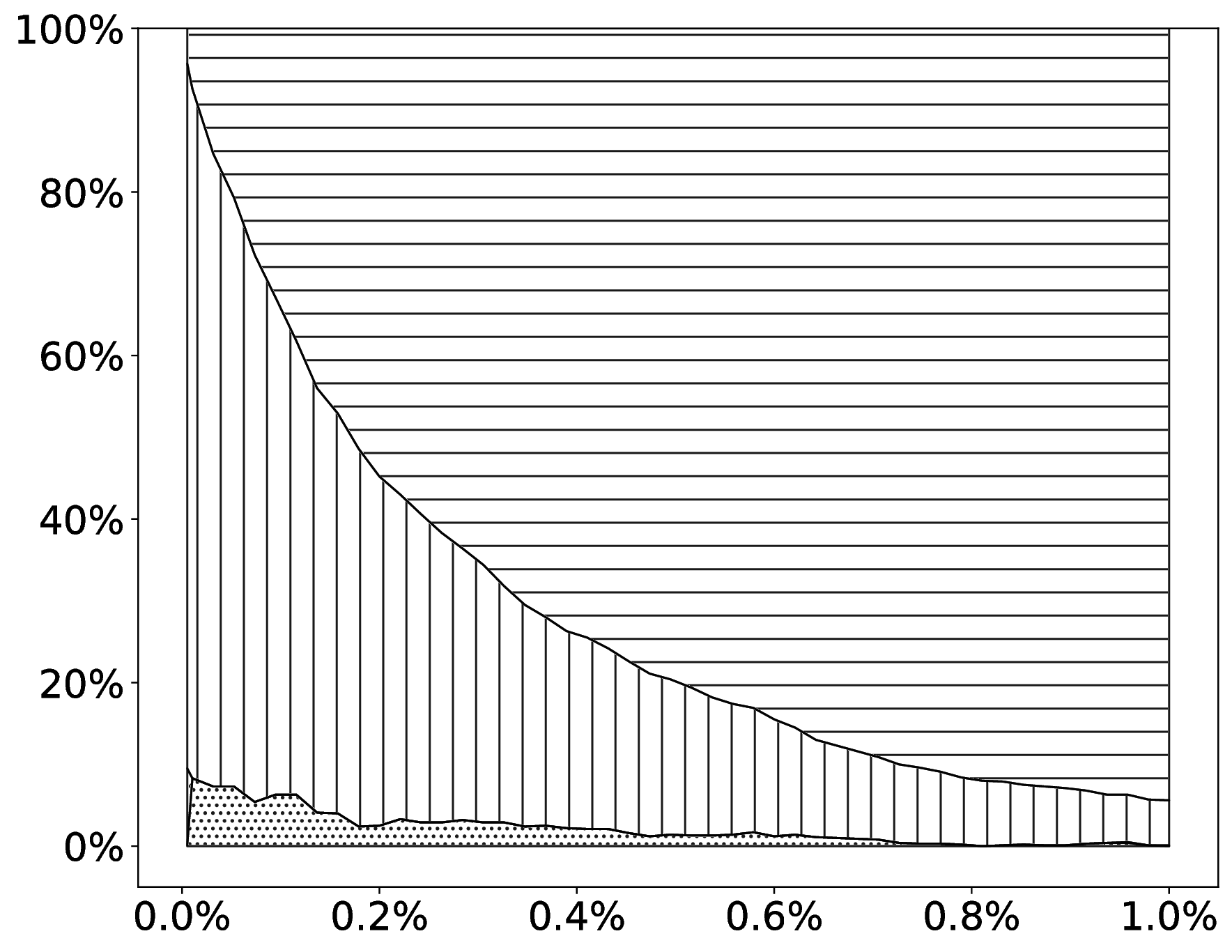}
         {\scriptsize (d) Letter}
          \end{subfigure}
       \hfill
         \begin{subfigure}[b]{0.23\textwidth}
         \centering
         \includegraphics[width=\textwidth]{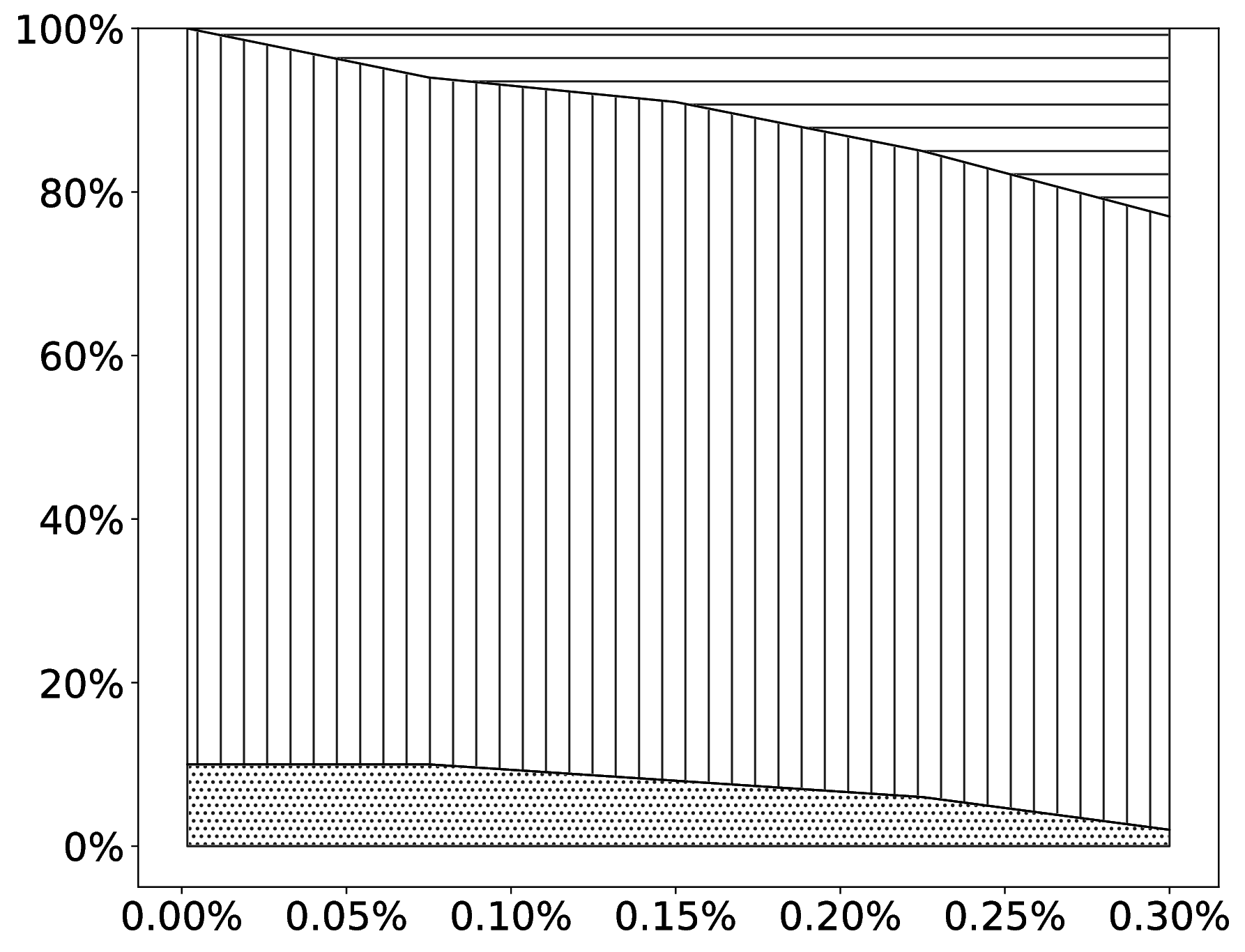}
          {\scriptsize {(e) MNIST}}
          \end{subfigure}
            \hfill
         \begin{subfigure}[b]{0.23\textwidth}
         \centering
         \includegraphics[width=\textwidth]{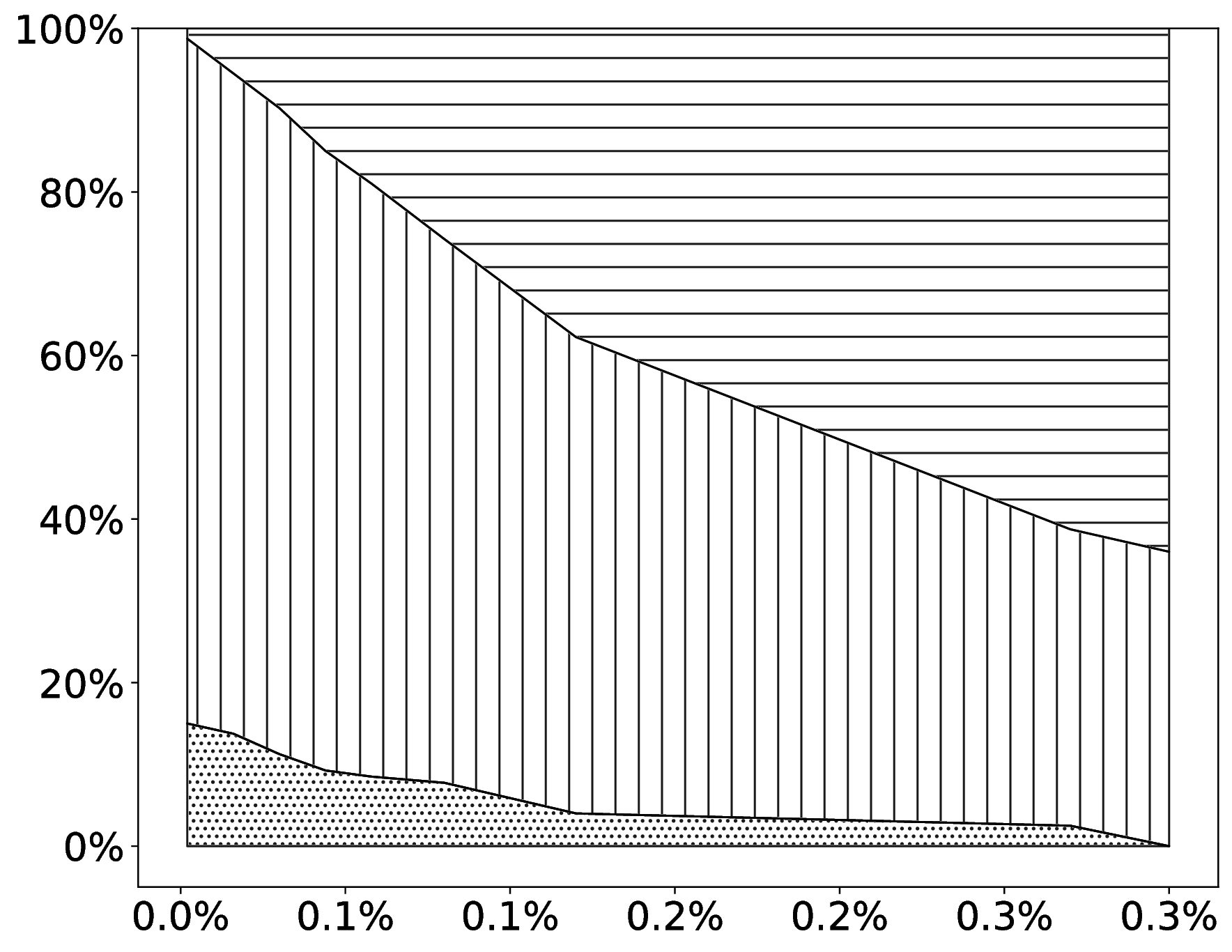}
         {\scriptsize (f) CIFAR10}
     \end{subfigure}
\caption{Results on how the poisoning threshold (in the $x$-axis) affects the percentages of certified, falsified, and unknown test cases (in the $y$-axis) in our method. Here, \emph{falsified} is in `$-$', \emph{unknown} is in `$.$', and \emph{certified} is in either `$|$' (quick certify) or `$/$' (slow certify). }
\label{fig:certified}
\end{figure}

Recall that in Algorithm~\ref{alg:our_method}, a test case may be certified in either Line~2 or Line~16.  When it is certified in Line~2, it belongs to the `$|$' region (quick certify) in Figure~\ref{fig:certified}.  When it is certified in Line~16, it belongs to the `$/$' region (slow certify).

For example, in Figure 6(e): When $n$=1, the falsify percentage is 0\%, the unknown percentage is 10\% and the quick-certify percentage is 90\%. When $n$=180, the falsify percentage is 23\%,  the unknown percentage is 2\%, and the quick-certify percentage is 75\%.

Figure~\ref{fig:certified} demonstrates the effectiveness of our method. Since the `$.$' regions that represent \emph{unknown} cases remain small, the vast majority of cases are successfully certified or falsified.

The results also reflect the nature of $n$-poisoning robustness:  as $n$ increases, the percentage of truly robust cases decreases. 
This is inevitable since having more poisoned elements in $T$ leads to  a higher likelihood of changing the classification label.  This is consistent with the results of prior studies \cite{shafahi2018poison,biggio2012poisoning,chen2017targeted}, which found that the prediction errors became significant even if a small percentage ($<0.2\%$) of training data in $T$ was poisoned. 


%

\section{Related Work}
\label{sec:related}

As explained earlier, while there has been prior work on certifying data-poisoning robustness for KNN, 
none of the existing methods can falsify the robustness property.
Thus, our method is the only one that can generate both certification and falsification results with certainty, and can handle both the learning and the prediction phases of a state-of-the-art KNN system.   In contrast, existing techniques such as Wang et al.~\cite{wang2018analyzing}, Jia et al.~\cite{jia2020certified,jia2020intrinsic}, and Weber et al.~\cite{weber2020rab} can only certify, but not falsify the robustness property.  Thus, in the presence of violations, these methods would remain inconclusive. Our method, on the other hand, can successfully resolve the robustness problem for most of the test inputs, as shown by our experimental evaluation.

KNN is not the only machine learning algorithm that is vulnerable to data poisoning. Other machine learning algorithms that are also found to be vulnerable to data poisoning include regression models~\cite{mei2015using}, support vector machines (SVM)~\cite{biggio2012poisoning,xiao2015support,xiao2012adversarial}, clustering algorithms~\cite{biggio2014poisoning}, and neural networks~\cite{shafahi2018poison,suciu2018does,demontis2019adversarial,zhu2019transferable}. So far, there has been no generic techniques for deciding the robustness property for all machine learning algorithms. 
Techniques have also been proposed to defend against data-poisoning attacks~\cite{steinhardt2017certified, tran2018spectral,   jagielski2018manipulating, feng2014robust, biggio2011bagging, bahri2020deep}, as well as to evaluate the effectiveness of defense  techniques~\cite{koh2018stronger, ma2019data} such as data sanitization~\cite{koh2018stronger} and differentially-private countermeasures~\cite{ma2019data}. 
Along this line, there is a growing interest in studying certified defenses~\cite{rosenfeld2020certified, levine2020deep,jia2020intrinsic} where robustness can be guaranteed either probabilistically or in a deterministic manner.

At a higher level, our method for using over-approximate analysis to narrow down the search space is analogous to static analysis techniques based on abstract interpretation~\cite{CousotC77}, which have been used to verify properties of both software programs~\cite{Wu019,SungKW17,Kusano016} and machine learning models~\cite{MohammadinejadP21,PaulsenWW20,PaulsenWWW20}, including robustness to data bias~\cite{MeyerAD21} and individual fairness~\cite{LiWWcav23}.
Furthermore, our method for detecting robustness violations is analogous to techniques used in bug-finding tools based on program verification and state space reduction~\cite{kroening2010verification,beyer2017software}. However, none of these techniques was designed to certify or falsify data-poisoning robustness of machine learning based systems.

Our method for using systematic testing to find robustness violations is related to the idea of fuzz testing~\cite{takanen2018fuzzing,miller1995fuzz} in the sense that mutations are used to generate violation-inducing inputs.
There is a large number of fuzz testing tools including AFL~\cite{MichalAFL},  honggfuzz~\cite{Honggfuzz}, libFuzzer~\cite{LibFuzzer}, SYMFUZZ~\cite{cha2015program}, and Driller~\cite{stephens2016driller}.
However, these tools focus primarily on search space pruning and search prioritization, e.g., by leveraging the syntax and semantics of the software code, but for KNN, the situation is significantly more complex.  This is because mutations of the training data can lead to drastic changes of the behavior of the underlying algorithm, during both the KNN inference phase and the KNN learning phase.
Thus, while existing techniques from the fuzz testing literature are inspiring, they are not directly applicable to this problem.

\section{Conclusion}
\label{sec:conclusion}

We have presented a method for deciding $n$-poisoning robustness accurately and efficiently for the state-of-the-art implementation of the KNN algorithm.   
To the best of our knowledge, this is the only method available for certifying as well as falsifying the complete KNN system, including both the learning and the prediction phases. 
Our method relies on novel techniques that first narrow down the search space using over-approximate analysis in the abstract domain, and then find violations using systematic testing in the concrete domain.
We have evaluated the proposed techniques on six popular supervised-learning datasets, and demonstrated the advantages of our method over two state-of-the-art techniques. 
\textcolor{black}{
Besides KNN, our method for over-approximating the impact of poisoning on the nearest neighbors is applicable to other distance-based machine learning classifiers and algorithms based on majority voting. Furthermore, since cross validation is a widely used parameter tuning technique in machine learning systems, our method for over-approximating cross validation is also applicable to other systems that rely on cross validation as a subroutine.
}

\begin{acks}
We thank the anonymous reviewers for their valuable feedback.
This work was partially funded by the U.S.\ NSF grants CNS-1702824 and CCF-2220345.
\end{acks}

\bibliographystyle{ACM-Reference-Format}
\bibliography{references}

\end{document}